\pdfoutput=1
\documentclass[journal]{IEEEtran}

\ifCLASSINFOpdf
\else
\fi
\usepackage[cmex10]{amsmath}
\DeclareMathOperator*{\argmax}{arg\,max}

\usepackage{tikz}
\usetikzlibrary{calc,arrows,positioning,intersections}
\usetikzlibrary{decorations.pathreplacing}
\usepackage{circuitikz}

\usepackage{graphicx}
\makeatletter
\newcommand*\bigcdot{\mathpalette\bigcdot@{1}}
\newcommand*\bigcdot@[2]{\mathbin{\vcenter{\hbox{\scalebox{#2}{$\m@th#1\bullet$}}}}}
\makeatother
\usepackage{cite}
\usepackage{subfigure}

\usepackage{bbding}
\usepackage{amssymb}
\usepackage{bbold}
\usepackage{varwidth}

\usepackage{array,booktabs} 
\newcolumntype{P}[1]{>{\centering\arraybackslash}p{#1}}
\newcolumntype{M}[1]{>{\centering\arraybackslash}m{#1}}

\usepackage{mathtools}
\usepackage{amsthm}
\newtheorem{proposition}{Proposition}
\newtheorem{theorem}{Theorem}

\usepackage[linesnumbered,ruled,vlined]{algorithm2e}
\makeatletter
\newcommand{\nosemic}{\renewcommand{\@endalgocfline}{\relax}}
\newcommand{\dosemic}{\renewcommand{\@endalgocfline}{\algocf@endline}}
\let\oldnl\nl
\newcommand{\nonl}{\renewcommand{\nl}{\let\nl\oldnl}}
\makeatother

\usepackage{MnSymbol}

\usepackage{algpseudocode}
\usepackage{caption}

\usepackage{amsfonts}
\usepackage{mathtools}

\usepackage{mathtools} 

\mathchardef\Re="023C
\mathchardef\Im="023D

\usepackage{dblfloatfix}

\usepackage{calligra}
\usepackage[T1]{fontenc}
\usepackage{mathrsfs}

\SetKwRepeat{Do}{do}{while}%

\usepackage[none]{hyphenat}
\usepackage{multirow}

\begin{document}

\title{Space-Time Signal Design for Multilevel Polar Coding in Slow Fading Broadcast Channels}

\author{Hossein~Khoshnevis,~\IEEEmembership{Member,~IEEE,}
        Ian~Marsland,~\IEEEmembership{Member,~IEEE,}
        Hamid~Jafarkhani,~\IEEEmembership{Fellow,~IEEE,}
        and~Halim~Yanikomeroglu,~\IEEEmembership{Fellow,~IEEE}
\thanks{H. Khoshnevis was with the Department of Systems and Computer Engineering, Carleton University, Ottawa, ON, Canada, when this work was done; he is now with the Huawei Canada Research Centre, Ottawa, ON. I. Marsland and H. Yanikomeroglu are with the Department of Systems and Computer Engineering, Carleton University, Ottawa, ON, Canada. H. Jafarkhani is with the Center for Pervasive Communications and Computing, University of California, Irvine, CA, USA. (e-mail: \{khoshnevis, ianm, halim\}@sce.carleton.ca, hamidj@uci.edu). 

This work is supported in part by Huawei Canada Co., Ltd., and in part by the Natural Sciences and Engineering Research Council of Canada's (NSERC) Strategic Partnership Grants for Projects (SPG-P) program (H. Khoshnevis and H. Yanikomeroglu), and in part by the NSF Award CCF-1526780 (H. Jafarkhani). This work was presented in part in \textit{IEEE PIMRC} 2017.}}

\maketitle

\begin{abstract}
Slow fading broadcast channels can model a wide range of applications in wireless networks. Due to delay requirements and the unavailability of the channel state information at the transmitter (CSIT), these channels for many applications are non-ergodic. The appropriate measure for designing signals in non-ergodic channels is the outage probability. In this paper, we  provide a method to optimize STBCs based on the outage probability at moderate SNRs. 

Multilevel polar coded-modulation is a new class of coded-modulation techniques that benefits from low complexity decoders and simple rate matching. In this paper, we derive the outage optimality condition for multistage decoding and propose a rule for determining component code rates. We also derive an upper bound on the outage probability of STBCs for designing the set-partitioning-based labelling. Finally, due to the optimality of the outage-minimized STBCs for long codes, we introduce a  novel method for the joint optimization of short-to-moderate length polar codes and STBCs.
\end{abstract}


\IEEEpeerreviewmaketitle

\section{Introduction}

\IEEEPARstart{I}{n} broadcast channels the  channel state information (CSI) varies from user to user and is generally not available at the transmitter. As such, spectrally efficient adaptive modulation and coding cannot be employed. Instead, the system should provide reliable communication to as many users as possible, at as high a rate as possible. To achieve reliable communication, the system can be designed to minimize the average frame error rate (FER) at an average signal-to-noise ratio (SNR) for given channel statistics (e.g., Rayleigh fading). Alternatively, the outage probability can be minimized since it is a lower bound on the FER of the system in non-ergodic slow fading broadcast channels \cite[references therein]{Jafarkhani2005}. 

Space-time block codes (STBCs) are a class of low complexity multiple-input multiple-output (MIMO) schemes that can achieve low outage probability without CSI at the transmitter. Thus, STBCs are a reasonable solution for communicating over slow fading broadcast channels. Orthogonal STBCs (OSTBCs), introduced by Alamouti in \cite{Alamouti1998} and by Tarokh, Jafarkhani and Calderbank in \cite{Tarokh1999-1}, benefit from low decoding complexity and can provide the full spatial diversity. However, they cannot achieve any coding gain and suffer from a rate loss as the number of antennas grows. To overcome these limitations, quasi-orthogonal STBCs \cite{Jafarkhani2001}, super-orthogonal space-time trellis codes \cite{Jafarkhani2003}, algebraic codes \cite{Belfiore2005,Sezginer2007},\cite[references therein]{Mietzner2009}, and space-time super-modulations \cite{Nikitopoulos2017} have been proposed.  

Tarokh et al., in \cite{Tarokh1998}, introduced the rank and determinant criteria as two useful measures for designing STBCs at high SNRs. Even though most STBCs are designed on the basis of these criteria, they do not guarantee good performance at low-to-moderate SNRs  \cite{Falou2012}. Instead, STBCs can be  optimized to achieve a specific outage probability at the lowest possible SNR.

 Polar codes, introduced by Erdal Ar{\i}kan, are a low complexity class of forward error correction (FEC) codes that work based on the concept of channel polarization as  a method to improve the reliability of some bit-channels at the expense of others \cite{Arikan2009}. Polar codes are uniquely designed for specific channel statistics and a given modulation scheme by determining the set of bit-channels that are used to carry the message (the information set).

For matching binary codes to the modulation and STBC, an efficient coded-modulation scheme should be used. It has been shown that polar coded-modulation, constructed on the basis of multilevel coding (MLC) with multistage decoding (MSD), outperforms polar coded-modulation constructed on the basis of bit-interleaved coded-modulation (BICM) \cite{Seidl2013-0}. This is due to clarity of design and the conceptual similarity of MLC with the set-partitioning-based bit-to-symbol mapping (SPM) to channel polarization, observed initially in \cite{Arikan2006}. Moreover, multilevel polar coded-modulation (MLPCM) outperforms BICM-based convolutional and turbo coded-modulation schemes as well \cite{khoshnevis2017-PIMRC1}. Indeed, MLPCM can provide a low complexity power-efficient scheme that can be employed in a wide range of wireless applications.

Typically, due to the fact that MLC/MSD achieves the channel capacity, the symbol-wise average mutual information is maximized to design the signal constellation for MLC/MSD \cite{Barsoum2007} assuming that capacity-achieving FEC codes are employed. Similarly, as we will show in this paper, the symbol-wise outage probability can be minimized to achieve reliable STBCs for non-ergodic channels.

To design the SPM, direct evaluation of the open-form measures on the performance of MSD, such as the sum of the binary channel cutoff rates \cite{Arikan2006}, is difficult. Instead, typically channel dependent pairwise measures are used to design set-partitioning \cite{Balogun2017}. One relevant pairwise measure for slow fading channels is the pairwise outage probability \cite{Vojcic1994}. Fortunately, by substituting the cutoff rate \cite{Massey1974} instead of the mutual information, a closed-form upper bound on the pairwise outage probability can be derived. 

The signal design based on the outage probability can improve the performance especially for long FEC codes \cite{Duyck2013precoding}. However, since it does not consider the structure of the FEC codes and decoders, it is not the best measure for designing signals used with short to moderate length codes. 

In this paper, we aim to enhance the performance of the concatenation of STBCs and multilevel polar codes for slow fading broadcast channels.  The main contributions of this paper are summarized as follows:

\begin{itemize}
\item{We propose a method to optimize STBCs, low complexity space block codes (SBCs) \cite{Hochwald2003,Mroueh2012}, and time-varying SBCs (TVSBCs) \cite{Duyck2012time1}, by minimizing the outage probability;}
\item{derive an outage rule to determine the component code rates of MLC/MSD in slow fading channels;}
\item{design the SPM for MLPCM based on a novel bound derived on the pairwise outage probability and a proposed algorithm to modify existing measures; and}
\item{propose a method for the joint optimization of polar codes and STBCs.}
\end{itemize}

The rest of the paper is organized as follows: in Section~\ref{sec:sysmodel}, the system model is defined; in Section~\ref{sec:classSTBC}, the STBC design methods and codes used in this paper are reviewed; in Section~\ref{sec:STBCoutage}, the STBC design by minimizing the outage probability is described;  in Section~\ref{sec:codedesign}, the design elements of the MLPCM, including the outage rule for determining the component code rates, the labelling algorithm, and the MLPCM design procedure are discussed; in Section~\ref{sec:JointOptimization}, the joint optimization of polar codes and STBCs is explained; in Section~\ref{sec:numresult}, the numerical results are presented; and in Section~\ref{sec:Conclusion}, the conclusions are provided.

\section{System Model}
\label{sec:sysmodel}

The system  consists of a transmitter and a receiver equipped with $N_{\rm{t}}$ transmit and $N_{\rm{r}}$ receive antennas. Each $K$ bits of data are coded using a multilevel binary polar code with a rate of $R_{\rm{tot}}=K/N_{\rm{tot}}$ and a length of $N_{\rm{tot}}=NB$ consisting of $B$ levels each with a  code length of $N$. Each level encoder of the multilevel code encodes a portion of the total $K$ bits corresponding to the component code rates $\{R_1,...,R_B\}$. After encoding all levels, each set of code bits \{$c_n^1,c_n^2,...,c_n^B$\}, for $n=1,...,N$, are mapped to a space-time symbol  by employing a multidimensional SPM. The space-time symbol is one of the signal points of a STBC, $\textbf{G}$, distributed on $L$ time slots and $N_{\rm{t}}$ transmit antennas. The space-time symbol is then sent through the $N_{\rm{t}} \times N_{\rm{r}}$  quasi-static MIMO Rayleigh flat fading channel $\textbf{H}$ with the distribution $\mathcal{CN}(\textbf{0},\textbf{I})$. The $L \times N_{\rm{r}}$ received samples are
\begin{equation}
\label{recievedsignal}
\textbf{Y}_n=\textbf{S}_n \textbf{H}+ \textbf{W}_n,
\end{equation}
where $\textbf{S}_n$ is the space-time signal and $\textbf{W}_n$ is the zero-mean complex additive white Gaussian noise (AWGN) with variance $N_0/2$ per dimension. Note that each MLPCM codeword only observes one independent realization of $\textbf{H}$. In the following we consider the transmission of the $n^{\rm{th}}$ space-time symbol, and drop the subscript $n$ to simplify the notation. The probability density function of the received  samples, $\textbf{Y}$, given perfect CSI at the receiver, for the space-time symbol $\textbf{S}$ with elements $s_{n_{\rm{t}}}^l$ being the complex constellation symbol transmitted in time slot $l$ from antenna $n_{\rm{t}}$, is given by 
 \begin{equation}
\label{prob}
\fontsize{9.5pt}{12pt}
\newcommand\abs[1]{\left|#1\right|^2}
p(\textbf{Y} \mid \textbf{S},\textbf{H})=\frac{1}{{(\pi N_0)}^{LN_{\rm{r}}}}\textrm{exp}\left(  \frac{-||\textbf{Y-SH}||_{\rm{F}}^2}{N_0}  \right),
 \end{equation}
where $||\cdot||_{\rm{F}}$ is the Frobenius norm.  Throughout this paper, the average SNR is adjusted using $N_0$ and we set ${\rm{E}}[||\textbf{S}||_{\rm{F}}^2]=1$. The STBC points are chosen using a multidimensional SPM that spans all bits of all STBC elements $s_{n_{\rm{t}}}^l$. This method of employing multidimensional SPM is elaborated on in \cite{Balogun2017} for the Golden code \cite{Belfiore2005} and the Grassmannian constellation. The received samples of codewords \{$\hat{\textbf{c}}^1,\hat{\textbf{c}}^2,...,\hat{\textbf{c}}^B$\} are decoded using an MSD, where $\hat{c}^b$ is the decoded codeword of the $b^{\rm{th}}$  level code.

In the MSD architecture, the code bits of each level are deduced with the aid of the received symbol and previously deduced code bits of the upper levels \cite{Imai1977}. At each level after decoding, the received message word is fed to an encoder and the generated codeword $\hat{\textbf{c}}^b$ reduces part of the ambiguity for the demapper of the next level, which enables it to compute more reliable log-likelihood-ratios (LLRs). Therefore, the LLR estimation at $b^{\rm{th}}$ level can be given by
$\lambda_{b}= \ln \frac{p(\textbf{Y}|\textbf{c}^{1:b-1}=\hat{\textbf{c}}^{1:b-1},c^b=0,\textbf{H})}{p(\textbf{Y}|\textbf{c}^{1:b-1}=\hat{\textbf{c}}^{1:b-1},c^b=1,\textbf{H})}$, where
\begin{equation}
\label{probabilityzeroone}
\begin{split}
 & p(\textbf{Y}|\textbf{c}^{1:b-1}=\hat{\textbf{c}}^{1:b-1},c^b=0,\textbf{H}) \\
 &= \frac{1}{2^{B-b+1}} \sum_{\textbf{c}^{b+1:B} \in [0,1]^{B-b+1}} p(\textbf{Y}|\textbf{S}=\text{SM}[\hat{\textbf{c}}^{1:b-1},0,\textbf{c}^{b+1:B}],\textbf{H}), \\
& p(\textbf{Y}|\textbf{c}^{1:b-1}=\hat{\textbf{c}}^{1:b-1},c^b=1,\textbf{H}) \\
 & =\frac{1}{2^{B-b+1}} \sum_{\textbf{c}^{b+1:B} \in [0,1]^{B-b+1}} p(\textbf{Y}|\textbf{S}=\text{SM}[\hat{\textbf{c}}^{1:b-1},1,\textbf{c}^{b+1:B}],\textbf{H}), \\
\end{split}
\end{equation}
and $\textbf{c}^{i:j}$ represents code bits levels $i$ to $j$ at a specific $n$ and $\text{SM}[\textbf{c}]$ defines the mapping from code bits $\textbf{c}^{1:B}$ to a space-time symbol according to the bit-to-symbol mapping rule and the STBC used. Note that there are various techniques to simplify the LLR calculation for STBCs (e.g. \cite{khoshnevis2019}).

\section{Review of STBC Design Methods}
\label{sec:classSTBC}

In this section, we review the structure and the design methods of STBCs, SBCs, and TVSBCs used in this paper.

\subsection{Space-time Block Codes} 

The pairwise error probability between two space-time symbols of a STBC, $\textbf{S}_i$ and $\textbf{S}_j$, can be bounded as \cite{Jafarkhani2005}
\begin{equation}
\label{PEP1}
{p(\textbf{S}_i  \rightarrow \textbf{S}_j|\textbf{H})} \leq \frac{1}{2} \text{exp}\Big(-\frac{||\Delta_{i,j}\textbf{H}||_{\rm{F}}^2}{4N_0}\Big),
\end{equation} 
where  $\Delta_{i,j}=\textbf{S}_i-\textbf{S}_j$  with elements $\delta_{l,n_{\rm{t}}}$. The rank and determinant criteria are extremely effective in minimizing the pairwise error probability of STBCs at high SNRs \cite{Tarokh1998}. The rank of the pairwise difference matrix, $\Delta_{i,j}$, as a measure of the diversity, and the determinant of $\Delta_{i,j}\Delta_{i,j}^{\mathcal{H}}$, where $\mathcal{H}$ is the Hermitian transpose, as a measure of the coding gain for STBCs have been widely used for the design of codes. One of the well-known STBCs designed with this method, capable of achieving full diversity and the highest coding gain for a $2 \times 2$ antenna configuration, is the Golden code introduced by Belfiore et al. in \cite{Belfiore2005}. The structure of the Golden code, herein referred to as Matrix A, is given by
\begin{equation}
\label{code1}
      \textbf{G}_{\rm{A}}  = \frac{1}{\sqrt{5}}
     \begin{bmatrix}
       \alpha(s_1+s_2\theta)  &   \alpha(s_3+s_4\theta)  \\
       \jmath\bar{\alpha}(s_3+s_4\bar{\theta})  &   \bar{\alpha}(s_1+s_2\bar{\theta})\\
     \end{bmatrix},
\end{equation} 
where $s_1$, $s_2$, $s_3$, and $s_4$  are complex constellation symbols of the code chosen from a constellation with cardinality $|2^{\hat{B}}|$, $\theta=(1+\sqrt{5})/2$, $\bar{\theta}=1-\theta$, $\alpha=1+\jmath(1-\theta)$, and $\bar{\alpha}=1+\jmath(1-\bar{\theta})$, and $\jmath=\sqrt{-1}$. Note that by substituting constellation points in $\textbf{G}_{\rm{A}}$, space-time symbols $\textbf{S}$ in (\ref{recievedsignal}) are generated.  Later, Sezginer and Sari in \cite{Sezginer2007} introduced a $2 \times 2$ STBC to reduce the decoding complexity of the Golden code. This code, herein referred to as Matrix B and can be written as 
\begin{equation}
\label{code2}
      \textbf{G}_{\rm{B}}  =
     \begin{bmatrix}
          \alpha_1s_1+\alpha_2s_3  &    \alpha_1s_2+ \alpha_2s_4  \\
         -\beta_1s^*_2-\beta_2s^*_4   &   \beta_1s^*_1+\beta_2s^*_3 \\
     \end{bmatrix},
\end{equation} 
where $\alpha_1=\beta_1=1/\sqrt{2}$, $\alpha_2=1/\sqrt{2}e^{\jmath\phi}$, $\beta_2=-\jmath\alpha_2$ and $*$ denotes the conjugate. By performing a numerical search, $\phi=114.29^{\circ}$ is found to maximize the minimum determinant for QPSK  and larger size QAM constellations \cite{Sezginer2007,Sezginer2009}. As shown in \cite{Sezginer2007}, $\textbf{G}_{\rm{B}} $ performs only slightly worse than $\textbf{G}_{\rm{A}} $.

 To approach the capacity,  the average mutual information can be used as an SNR-dependent measure for designing and analyzing STBCs \cite{Hassibi2002}. The mutual information for a given realization of $\textbf{H}$ can be written as 
\begin{equation}
\label{ch9:mutualinfo}
I(\textbf{Y};\textbf{S}|\textbf{H})=\sum_{i=1}^{2^B}\text{Pr}(\textbf{S}_i){\rm{E}}_{\textbf{Y}}\Big[  \log_2 \Big( \frac{p(\textbf{Y}|\textbf{S}_i,\textbf{H})}{\sum_{j=1}^{2^B}\text{Pr}(\textbf{S}_j) p(\textbf{Y}|\textbf{S}_j,\textbf{H})} \Big) \Big].
\end{equation}
For notational convenience, $I(\textbf{Y};\textbf{S}|\textbf{H})$ is denoted as $I$ hereinafter. The average mutual information can be estimated by taking the expectation over $\textbf{H}$ as $I(\textbf{Y};\textbf{S})= {\rm{E}_{\textbf{H}}}[I]$.

In \cite{Falou2012}, to achieve good performance at low SNRs, the Matrix B code is optimized by maximizing the mutual information at low SNRs. As a result, the parameters of the Matrix B code for low SNRs are different from the code designed based on rank and determinant criteria. For example, for a wide range of low-to-moderate SNRs and by employing QPSK, $\phi$ is found to be $135^{\circ}$ for $N_{\rm{r}}=2$.

\subsection{Space Block Codes and Time-Varying Space Block Codes}
\label{SBCcondition}

Most high-performance STBCs have unnecessarily high decoding complexity.  Instead, a strong binary FEC outer code used with a simple SBC can achieve most of the coding gain of a complex STBC that is achieved by temporal expansion of the signal. Therefore, to reduce the complexity of STBC detection, the code can be expanded only across the antennas (i.e. $L=1$) to increase the transmission rate compared to single input multiple output systems. This class of STBCs is herein referred to as SBCs. In this section, we review the structure and design method of SBCs. The simplest form of SBCs is the vector channel symbol introduced by Hochwald and ten Brink in \cite{Hochwald2003}. This code can be written as $\textbf{G}_{\rm{C}} =[s_1  \quad s_2 \quad ... \quad s_{N_{\rm{t}}}]$ where $s_{n_{\rm{t}}}$ is chosen from a quadrature amplitude modulation (QAM). Instead of using available constellations in the structure of SBCs, parameterized SBCs can be optimized according to channel statistics. A parameterized SBC introduced in \cite{vakilian2015high}, herein referred to as $\textbf{G}_{\rm{D}} $, can be written as $\textbf{G}_{\rm{D}} =[\alpha_1  s_1 + \beta_1 s_2  \quad  \alpha_2 s_1+\beta_2 s_2]$. Using $\textbf{G}_{\rm{D}} $, the received  vector can be given as $\textbf{y}=\textbf{sVH}+\textbf{w}$ where $\textbf{s}=[s_1 \quad s_2]$ and $\textbf{V}$ is given as
\begin{equation}
\label{code1}
      \textbf{V}  =
     \begin{bmatrix}
       \alpha_1 &  \alpha_2  \\
       \beta_1  &  \beta_2  \\
     \end{bmatrix}.
\end{equation} 
To achieve the maximum diversity for $\textbf{G}_{\rm{D}} $, $\textbf{sVH}$ must be of full rank. Thus, the determinant of $\textbf{H}$ and $\textbf{V}$ should be non-zero. We know that $|\textbf{H}|\neq 0$  since any matrix with random variable elements will be full rank with probability one \cite{Jafarkhani2005}.

In contrast, $\textbf{V}$ should be adjusted to have a non-zero determinant. This results in $\alpha_1\beta_2 \neq \alpha_2\beta_1$ which must be used when designing $\textbf{G}_{\rm{D}} $. In $\textbf{G}_{\rm{D}} $ only two symbols are employed. By increasing the number of symbols, the number of degrees of freedom for the optimization increases. To this end, we propose to use four different symbols in the structure of a SBC. The corresponding SBC structure can be given as $\textbf{G}_{\rm{E}} =[\alpha_1  s_1 + \beta_1 s_2  \quad  \alpha_2 s_3+\beta_2 s_4]$. The matrix $\textbf{V}$ for $\textbf{G}_{\rm{E}} $ can be written as
\begin{equation}
\label{code1}
      \textbf{V}  =
     \begin{bmatrix}
       \alpha_1 &  \beta_1 & 0 & 0 \\
        0 & 0 & \alpha_2  &  \beta_2  \\
     \end{bmatrix}^\mathcal{T},
\end{equation} 
where $\mathcal{T}$ is the transpose operator. In this case, as soon as all coefficients are non-zero, the matrix $\textbf{V}$ is of full rank. 

In the next sections, we optimize the SBCs using different measures. All these codes are designed on the basis of the following principles. First, to limit the emitted power of each antenna, we assume $|\alpha_1|^2+|\beta_1|^2 = |\alpha_2|^2+|\beta_2|^2$.  Next, to limit the search space, we set $\alpha_2=\alpha_1$ and $\measuredangle \alpha_1=0^{\circ}$. To satisfy  $\alpha_1\beta_2 \neq \alpha_2\beta_1$ for $\textbf{G}_{\rm{D}} $, we can set $\beta_2=\jmath\beta_1$. Without loss of generality, the same condition can be used for $\textbf{G}_{\rm{E}} $. Therefore, the phase of $\beta_1$, $\varphi=\measuredangle \beta_1$, and the ratio of magnitudes $|\beta_1|/|\alpha_1|$  should be optimized. To keep the total power constant, a power constraint is considered which can be given as $\frac{1}{N_{\rm{t}}L}\sum_{n_{\rm{t}}} \sum_{l} |s_{n_{\rm{t}}}^l|^2=1$.

 SBCs with a higher number of transmit antennas can substantially enhance the performance. In this paper, we also optimize the extension of $\textbf{G}_{\rm{E}} $ for $N_{\rm{t}}=3$ given as $\textbf{G}_{\rm{F}}=[\alpha_1  s_1 + \beta_1 s_2  \quad  \alpha_2 s_3+\beta_2 s_4 \quad  \alpha_3 s_5+\beta_3 s_6]$. For $\textbf{G}_{\rm{F}}$, to limit the search space, we set $\alpha_1=\alpha_2=\alpha_3$ equal to a constant and optimize the other parameters given the power condition $|\alpha_1|^2+|\beta_1|^2 = |\alpha_2|^2+|\beta_2|^2= |\alpha_3|^2+|\beta_3|^2$.

Introduced by Duyck et al. in \cite{Duyck2012time1}, TVSBCs can improve the performance of SBCs by providing a wide range of rotations for symbols of each antenna during a codeword transmission. This in turn results in limiting the effect of the worst-case fading rotation and increasing the diversity. For implementing the TVSBC, the corresponding SBC can be multiplied by a $N_{\rm{t}} {\times} N_{\rm{t}}$  diagonal matrix $\textbf{A}_n$ which is known at the receiver. The main diagonal elements of $\textbf{A}_n$ are time-varying (TV) random complex numbers with constant magnitudes given as $e^{\jmath \theta_i(n)}$. Therefore, the new TVSBC can be written as $\textbf{G}(n)=\textbf{G}\textbf{A}_n$. Since approximately universal codes are already designed to be robust against the worst-case rotation \cite{Tavildar2006}, unlike SBCs, they cannot be improved by applying time-varying schemes.

\section{STBC Design based on the Outage Probability}
\label{sec:STBCoutage}

For non-ergodic channels, the outage probability is an achievable bound on the FER of an FEC coded system \cite{Jafarkhani2005}. Thus, it is a useful measure for the design and analysis of STBCs used with outer channel coding. The outage probability is the probability of the instantaneous mutual information being less than a specific target rate $R_{\rm{tot}}B$. It can be written as 
\begin{equation}
\label{ch2:outage1}
P_{\rm{out}}(\textbf{S},R_{\rm{tot}}B)\coloneqq {\rm{Pr}}(I<R_{\rm{tot}}B)=\mathcal{F}(R_{\rm{tot}}B),
\end{equation}
where $\mathcal{F}(R_{\rm{tot}}B)$ is the cumulative distribution function (CDF) of $I$. For the sake of simplicity, the total outage probability of a STBC is denoted by $\epsilon$ hereinafter. Similarly, by substituting the mutual information for the $b^{\rm{th}}$ address-bit of a specific STBC in (\ref{ch2:outage1}), the level-wise outage probabilities, $\epsilon_b$, can be defined. The outage probability can be evaluated numerically by noting that
\begin{equation}
\label{ch2:outage}
\epsilon= \boldsymbol{\int}\mathbb{1}\big(I,R_{\rm{tot}}B\big) f_{\textbf{H}} d\textbf{H},
\end{equation}
where the distribution of $\textbf{H}$ is given as $f_{\textbf{H}}= \frac{1}{\pi^{N_{\rm{t}}N_{\rm{r}}}}\exp(-||\textbf{H}||_{\rm{F}}^2)$, and  $\mathbb{1}(u,v) $ is the unit step function defined as 
\begin{equation}
\label{ch2:outagecutoff6}
 \mathbb{1}(u,v) =
\begin{cases} 1  &  u<v, \\  0 & \text{otherwise}.\end{cases}
\end{equation}

As a parallel concept, the $\epsilon$-outage capacity for a given outage probability $\epsilon$ is defined as 
\begin{equation}
\label{ch2:outagecapacity}
C_{\epsilon}=\mathcal{F}^{-1}(\epsilon).
\end{equation}
Similarly, $C_{\epsilon_b,b}$ can be defined for different levels of a space-time signal. For a given $\epsilon_b$,  $C_{\epsilon_b,b}$ can be approximated numerically by estimating the outage probability for a large number of  target rate values between 0 and 1 and choosing the one corresponding to the target $\epsilon_b$.

 MLC/MSD used with STBCs optimized by minimizing the outage probability can achieve high performance since, as we show in the following theorem, the outage probability of MLC/MSD approaches the constellation-constrained outage probability. 

\begin{theorem}
\label{Theorem11}
The outage probability of MLC/MSD scheme achieves the constellation-constrained outage probability of the slow fading broadcast channel.
\end{theorem}
\begin{proof}
Using the chain rule of the mutual information, we can connect the total mutual information and level-wise mutual informations as follows: 
\begin{equation}
\label{proof1}
\begin{split}
 \epsilon &={\rm{Pr}}(I<R_{\rm{tot}}B) ={\rm{Pr}}(I(\textbf{Y};\textbf{c}^{1:B}|\textbf{H})<R_{\rm{tot}}B) \\
 & ={\rm{Pr}}(\sum_{b=1}^B I(\textbf{Y};\textbf{c}^{b}|\textbf{c}^{1:b-1},\textbf{H})<R_{\rm{tot}}B) ={\rm{Pr}}(\sum_{b=1}^B I_b<R_{\rm{tot}}B),
\end{split}
\end{equation}
where $I_b$ is defined as the level-wise mutual information given by $I(\textbf{Y};\textbf{c}^{b}|\textbf{c}^{1:b-1},\textbf{H})$.
\end{proof}

To compute the outage probability, due to unavailability of closed-form expressions, (\ref{ch2:outage}) can be computed numerically or simulation can be used. In this paper, since the search space is continuous, the particle swarm optimization (PSO) \cite{Kennedy1995} is employed. For a detailed explanation of the steps, see \cite{Clerc2010} and the references therein. For  faster convergence of the algorithm, we modified the PSO according to the following principles. As the FER decreases, the number of realizations of $\textbf{H}$ to achieve an accurate estimation of the mutual information or the outage probability should be increased. When the PSO starts, a small number of realizations of $\textbf{H}$ is enough to achieve a coarse estimation of the outage probability for the initial population since they are chosen randomly and thus are most likely far from local or global minima. Therefore, we  increase the number of realizations of $\textbf{H}$ to estimate the mutual information or the outage as the best function value is lowered. For a detailed explanation about the optimization algorithm see \cite{khoshnevis2018}.

To design STBCs and SBCs using the outage probability, we use the conditions on shaping and power provided in Section~\ref{SBCcondition} to limit the search space and then we use the  modified PSO method to determine the rest of the parameters.  In this paper, by employing an additional linear search over the SNR \cite[Algorithm 2]{Balogun2017}, all STBCs and SBCs are optimized at $\epsilon=0.01$ for $N_{\rm{r}}=N_{\rm{t}}$. The optimized values for $\textbf{G}_{\rm{B}} $ for QPSK constellation and $R_{\rm{tot}}=1/2$ are given as $\alpha_1=\beta_1=0.314$, $\alpha_2=0.067 + 0.381\jmath$, $\beta_2= -0.070 + 0.384\jmath$. The optimized values of $\alpha_1$, $\beta_1$, and $\varphi$ for $R_{\rm{tot}}=1/2$  are as follows: for $\textbf{G}_{\rm{D}} $, 0.5, 0.5, and $330^\circ$; for TV $\textbf{G}_{\rm{D}} $, 0.5, 0.5, and $180^\circ$; for $\textbf{G}_{\rm{E}} $, 0.5, 0.5, and $45^\circ$; and for TV $\textbf{G}_{\rm{E}} $, 0.5, 0.5, and $45^\circ$. The  parameters $\alpha_1$, $\beta_1$, and $\varphi$ for $R_{\rm{tot}}=9/10$ are as follows: for $\textbf{G}_{\rm{D}} $, 0.5, 0.5, and $270^\circ$; for TV $\textbf{G}_{\rm{D}} $, 0.48, 0.52, and $120^\circ$; for $\textbf{G}_{\rm{E}} $, 0.33, 0.63, and $195^\circ$; and for TV $\textbf{G}_{\rm{E}} $, 0.34, 0.62, and $285^\circ$.

The outage probability results  of the outage-optimized schemes are shown in Fig.~\ref{ch9:fig:fig:outage_opt}. For low-to-moderate FEC code rates, TV $\textbf{G}_{\rm{E}} $ shows the minimum outage probability but performs very close to $\textbf{G}_{\rm{C}} $.  At high rates, for moderate values of the outage probability (e.g., $0.001$), $\textbf{G}_{\rm{E}} $ shows the least outage probability while for lower values of the outage probability, TV $\textbf{G}_{\rm{D}} $ is better. Indeed, the TV $\textbf{G}_{\rm{D}} $ benefits from a higher diversity, which is useful at high spectral efficiencies, while $\textbf{G}_{\rm{E}} $ benefits from more flexibility in optimization, which results in better codes for low-to-moderate spectral efficiencies. 

\begin{figure}
\centering  
\subfigure[]{\label{fig:a}\includegraphics[width=0.47\textwidth]{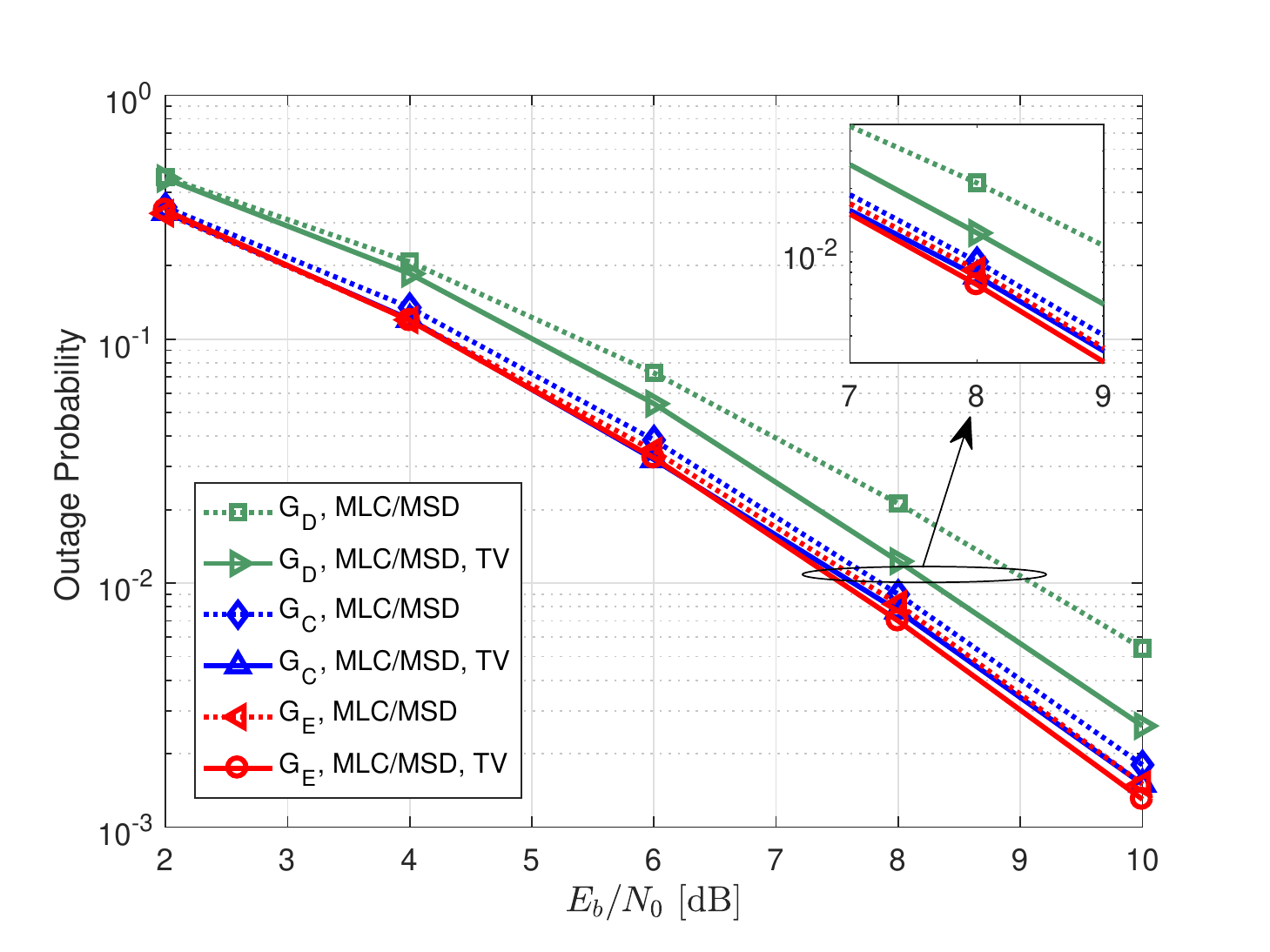}}
\subfigure[]{\label{fig:b}\includegraphics[width=0.47\textwidth]{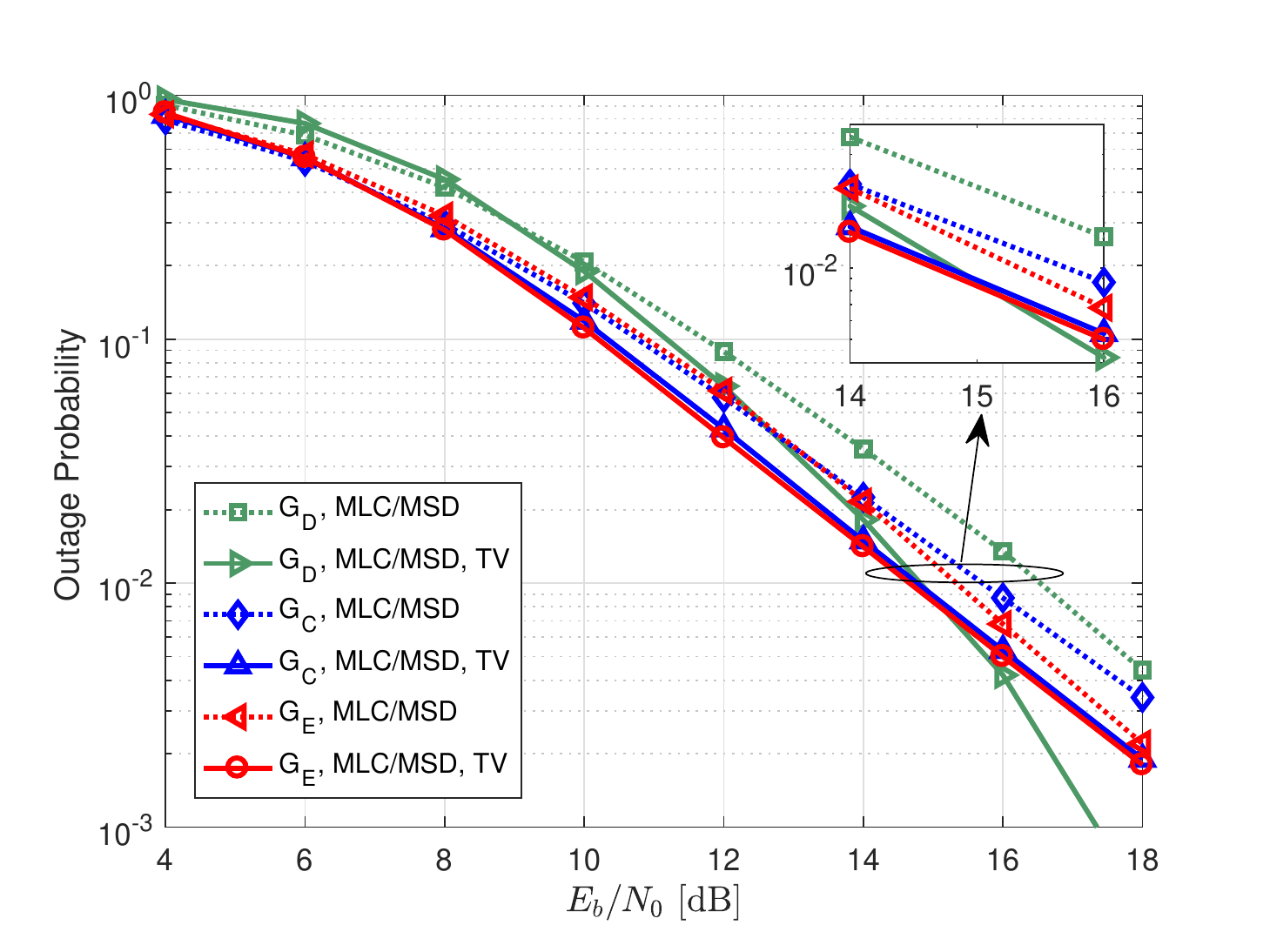}}
\caption{Outage probability comparison of $\textbf{G}_{\rm{C}} $, $\textbf{G}_{\rm{D}} $, and $\textbf{G}_{\rm{E}} $ and their time-varying variant with a) $R_{\rm{tot}}=0.5$ and 4 bits-per-channel-use (bpcu), and b) $R_{\rm{tot}}=0.9$ and 7.2 bpcu, both in a $2\times 2$ antenna configuration.}
\label{ch9:fig:fig:outage_opt}
\end{figure}

\section{Polar Coded-modulation Design}
\label{sec:codedesign}

In this section, we explain different aspects of polar coded-modulation design including determining component code rates, designing the SPM, and the coded-modulation design procedure.  To design the polar code, the information set should be determined. In this paper, we use the simulation-based best bit-channels rule proposed in \cite{Balogun2017}. In this method, we measure the number of first error events for the bit-channels of all levels and choose the bit-channels with the lowest error probabilities to determine $\{R_1,...,R_B\}$. 

\subsection{Determining Component Code Rates based on the Outage Probability}
\label{Componentcoderate}

In this section, we derive a rule based on the outage probability for determining component code rates and compare it with the best bit-channels rule. Two main design rules for determining component code rates of MLC are the capacity rule and the equal error probability rule for an ergodic channel \cite{Wachsmann1999}. When using the capacity rule, the component codes rates are chosen based on the constellation-constrained subchannel capacity. The capacity rule is only optimal for capacity achieving codes. However, for non-capacity achieving codes, the equal error probability rule is typically employed due to the possibility of deriving bounds on the performance of the coded-modulation scheme. For non-ergodic channels, the parallel concepts are the outage rule and the equal error rate rule. The equal error rate rule can be applied similarly to the ergodic channel. However, the outage rule needs to be derived.

When adjusting the level code rates based on the outage rule, we would like to set $R_b=C_{\epsilon_b,b}$ for the optimal $\epsilon_b$ values that minimize the total outage of the coded-modulation scheme using independent decoding of different levels. Let us denote the total outage probability of sequential decoding of levels as $P_{\rm{out}}^{\rm{SD}}$. This is a lower bound on the FER of the coded-modulation scheme.

\begin{theorem}
\label{Theorem12}
$P_{\rm{out}}^{\rm{SD}}$ is minimized if level-wise outage probabilities are equal and for any ${\normalfont  \textbf{H}}$  realization, either all levels are in outage or no level is in outage.
\end{theorem}
\begin{proof}
In MSD, an outage event occurs as soon as at least one level is in outage. Thus, we can write $P_{\rm{out}}^{\rm{SD}}=p( \displaystyle \bigcup_{b}\{I_b<R_b\})$. Because the probability of the events $\{I_b<R_b\}$ for all $b=1,...,B$ is never less than the maximum of the probabilities of individual events, we have $P_{\rm{out}}^{\rm{SD}} \geq \displaystyle \max_b {p( I_b<R_b)}$, with equality  if the events are fully dependent. Since we assumed for all realizations of $\textbf{H}$, either all levels are in outage or no level is in outage, the events are fully dependent and therefore, we have $P_{\rm{out}}^{\rm{SD}} = \displaystyle \max_b {p(I_b<R_b)}$. Now let us assume $C_{\epsilon_b,b}$ is chosen so that the level-wise outage probabilities are equal. Setting $R_1<C_{\epsilon_1,1}$ results in decreasing $p(I_1<R_1)$. However, since the total $\sum_bR_b$ is constant, we should set $R_2>C_{\epsilon_2,2}$ that results in increasing $p(I_2<R_2)$. Therefore, the maximum outage probability and consequently $P_{\rm{out}}^{\rm{SD}}$ increase. Thus, we should choose equal level-wise outage probabilities  in order to minimize $P_{\rm{out}}^{\rm{SD}}$.
\end{proof}

Note that in practice when $\epsilon_1=\epsilon_2=...=\epsilon_B$ for large FEC codes, we observe that either all levels are in outage or no level is in outage. Therefore, based on the results from Theorem~\ref{Theorem12}, to determine the level code rates $R_b$, we should minimize the level-wise outage probability $\epsilon_1=\epsilon_2=...=\epsilon_B$ such that  $\sum_bC_{\epsilon,b}=\sum_bR_b=R_{\rm{tot}}$. Thus, in an attempt to find the code rates, we can estimate the $C_{\epsilon_b,b}$s using (\ref{ch2:outagecapacity}) for a given $\epsilon_b$ and check whether  $\sum_bC_{\epsilon,b}=R_{\rm{tot}}$. Note that $\epsilon_b$ should be changed from zero to one with the step length $M$. However, the search space can be substantially limited by starting from $\epsilon_b=\epsilon$. In fact, the total outage of the joint decoding, $\epsilon$, is a lower bound on the outage of different levels since in Theorem~\ref{Theorem12} we assumed $P_{\rm{out}}^{\rm{SD}} = \displaystyle \max_b p(I_b<R_b)$ and $P_{\rm{out}}^{\rm{SD}}\rightarrow \epsilon$ only when $N\rightarrow\infty$.

Algorithm~\ref{AlgorithmRates} presents the outage rule for determining level code rates. In this algorithm, Function Estimate\_Mu(.) estimates the $b^{\rm{th}}$ level mutual information by substituting the probabilities of transmitting zero and one at each level corresponding to (\ref{probabilityzeroone}) in (\ref{ch9:mutualinfo}). In each iteration, the level-wise outage capacities $C_{\epsilon_b,b}$ are estimated. If $\sum_b C_{\epsilon_b,b}$ is less than the target spectral efficiency, the threshold $\hat{\epsilon}$ is increased by a factor $M$ and the same steps are repeated. Here, we set $M=1.05$. 
\begin{algorithm}
 \DontPrintSemicolon
     \SetKwInOut{Input}{Input}
     \SetKwInOut{Output}{Output} 
     \Input{$\textbf{G}_{\rm{X}}$, $R_{\rm{tot}}$}
     \Output{Component code rates $\textbf{R}$}
     \nonl   \textbf{Procedures}:  [${\textbf{I}}_{1:B}$]=Estimate\_Mu($\textbf{G}_{\rm{X}}$): Estimates the $1^{\rm{st}}$ to $B^{\rm{th}}$  level mutual information using (\ref{ch9:mutualinfo}) for $\hat{N}$ realizations of $\textbf{H}$ and outputs them in an $\hat{N} \times B$ matrix $\textbf{I}_{1:B}$. Find($\textbf{u}$,$v$): Outputs  $\displaystyle \argmax_i{u_i}<v$.\;
      [$\textbf{I}_{1:B}$]=Estimate\_Mu($\textbf{G}_{\rm{X}}$)\;
      $\hat{\epsilon}=\boldsymbol{\int_{\textbf{H}}} \mathbb{1}\big(I,R_{\rm{tot}}B\big) f_{\textbf{H}} d\textbf{H}$  \Comment{Using (\ref{ch2:outage})} \;
      \Do{$\sum_b C_{\epsilon_b,b}< R_{\rm{tot}}B$} {
      \lFor {b=1:B}{$C_{\epsilon}=\frac{1}{\hat{N}}$ Find($\textbf{I}_b$,$\hat{\epsilon}$)  \Comment{Using (\ref{ch2:outagecapacity}) \hspace{4.2cm}} }
       $\hat{\epsilon}=\hat{\epsilon}M$\;
       }
       $\textbf{R}=[C_{\hat{\epsilon},1},C_{\hat{\epsilon},2},...,C_{\hat{\epsilon},B}]$\;
       return $\textbf{R}$ \;
    \caption{{Outage Rule for Determining Component Code Rates}}
    \label{AlgorithmRates}
\end{algorithm}

Fig.~\ref{fig:outagerule} shows the comparison of the FERs of the equal FER rule, the outage rule, and the simulation-based best bit-channels rule. We observe that the FER of the equal FER rule is 1 dB worse than the other two rules and the FER of the best bit-channels rule is approximately the same as that of the outage rule. This shows that the simulation-based rule can predict the FERs correctly. For the rest of the paper, we use the simulation-based rule to determine the code rates.

\begin{figure}[h!]
\center
 \includegraphics[width=0.47\textwidth]{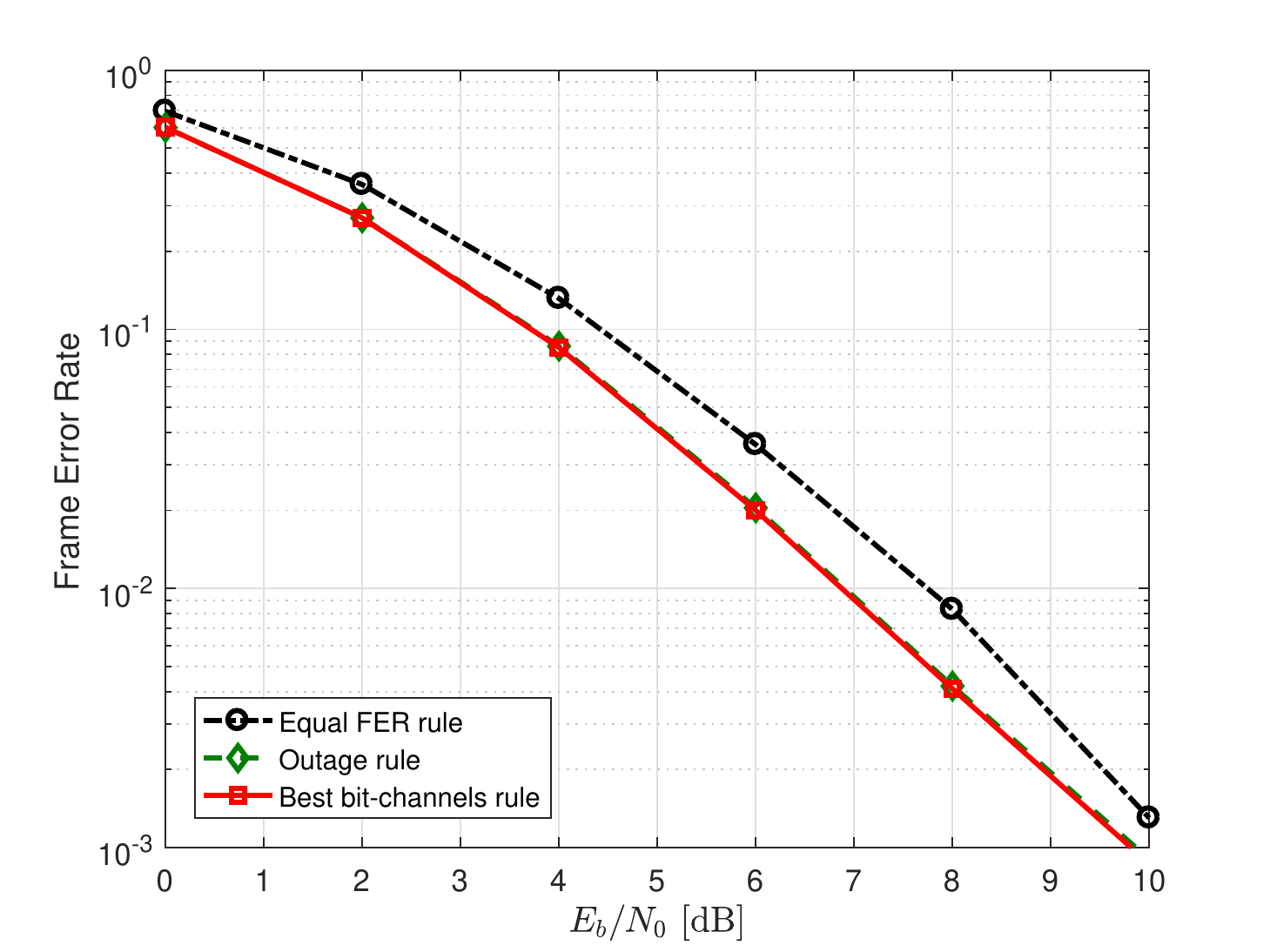}
\caption{FER comparison of polar code level rate determining algorithms, all for the Alamouti code with 2 bpcu and $N_{\rm{tot}}=1024$ in a $2\times 2$ antenna configuration.}
\label{fig:outagerule}
\end{figure}

\subsection{Labelling Algorithm}
\label{labelling}

Typically, set-partitioning algorithms with relevant channel measures are used to construct good bit-to-symbol mappings for MLC/MSD \cite{Ungerboeck1987}. Recently, for general irregular multidimensional constellations, in \cite{Balogun2017} a set-merging algorithm to create a set-partitioning is proposed that works as follows. At the first step, for each symbol, the most distant symbol is found and the minimum of them is considered as a distance threshold for pairing the sets denoted as $\tau$. Then every symbol is paired with the closest symbol with the minimum distance of $\tau$ and labels $0$ and $1$ are assigned to the first and the second symbols in each pair.  The threshold $\tau$ ensures that the set of distances of paired points does not have a high variance. After pairing every two symbols, the distance of each two-symbol set with respect to other two-symbol sets is measured and the minimum of the maximum of their distances is chosen as the new $\tau$ and the mentioned procedure is repeated in $B$ consecutive steps. In \cite{Balogun2017}, for fast fading channel, the Frobenius norm is used as a measure of distance between subsets. The SPM typically regularizes the component code rates in an incremental order, i.e., $R_1<R_2<...<R_B$. Note that the SPM indeed increases the variability between channel capacities. 

For the slow fading channel, based on the result of Theorem~\ref{Theorem12}, we are interested in equalizing the level-wise outage probabilities. If we assume that the code rates are in an incremental order, we expect that for any given fading realization, the average mutual information of levels is also incremental, i.e., $I_1<I_2<...<I_B$. Thus, we can employ the algorithm in \cite{Balogun2017} assuming that the appropriate measure is used. If we use approximately universal codes, the corresponding design measure is the determinant criterion when $N_{\rm{r}}>N_{\rm{t}}$ and the product of smallest $\min(N_{\rm{t}},N_{\rm{r}})$ singular value of $\Delta_{i,j}$ for $N_{\rm{r}}<N_{\rm{t}}$ \cite{Tavildar2006}. For SBCs, the distance is found to be the Frobenius norm using the pairwise error probability analysis according to \cite{Tarokh1998}.

For TVSBCs, fading coefficients change slowly while AWGN and TV sequences change fast. Therefore, we need a class of bounds that can model both effects to achieve good design measures. However, bounds proposed in \cite{Tarokh1998} cannot model both slow- and fast-changing parameters since they are derived under the assumption that the coefficients of a STBC remains constant during one codeword transmission. Alternatively, the pairwise outage probability\footnote{In \cite{Vojcic1994}, it is shown that the outage probability can be upper bounded by union bounds of pairwise outage probabilities. Thus, the pairwise outage probability is a relevant measure for the system behavior.} may be employed as a measure to design good SPM since it can model both kinds of parameters. 

Assuming that the cutoff rate is a lower bound on the mutual information, we can derive bounds on the outage probability. The cutoff rate, $R_0$, for a fixed STBC and \textit{a priori} probabilities, can be written as 
\begin{equation}
\newcommand\abs[1]{\left|#1\right|^2}
\label{cutoff12}
R_0(\textbf{S}|\textbf{H})= -\text{log}_2 \bigg\{  \sum_{i=1}^{2^B} \sum_{j=1}^{2^B} \hat{\rho}_{i,j} \bigg\},
\end{equation}
where $\hat{\rho}_{i,j} = \rho(\textbf{S}_i,\textbf{S}_j|\textbf{H}) \textrm{Pr}(\textbf{S}_i)\textrm{Pr}(\textbf{S}_j)$, in which $\rho(\textbf{S}_i,\textbf{S}_j|\textbf{H})$ is the pairwise Bhattacharyya coefficient,  defined as
\begin{equation}
\newcommand\abs[1]{\left|#1\right|^2}
\label{pairwise1}
{\rho(\textbf{S}_i,\textbf{S}_j|\textbf{H})= \boldsymbol{\int}} \sqrt{ {p(\textbf{Y}|\textbf{S}_i,\textbf{H}) } p(\textbf{Y}|\textbf{S}_j,\textbf{H})}d\textbf{Y}.
\end{equation}

 Substituting the cutoff rate, (\ref{ch2:outage}) can be upper-bounded as 
\begin{equation}
\label{ch2:outagecutoff}
\epsilon \leq p({R_0(\textbf{S}|\textbf{H})}<R_{\rm{tot}}B).
\end{equation}
Note that since the cutoff rate is a lower bound on the mutual information, the rate region of the outage is larger and thus (\ref{ch2:outagecutoff}) is an upper bound on the outage probability. The bound can be written as
\begin{equation}
\label{ch2:cutoffboundoutage}
p(R_0({\textbf{S}|\textbf{H})}<R_{\rm{tot}}B)= p \big( 2^{-R_{\rm{tot}}B}<\sum_{i=1}^{2^B}\sum_{j=1}^{2^B}\hat{\rho}_{i,j} \big).
\end{equation}

 Similar to (\ref{ch2:cutoffboundoutage}), we derive a closed-form for the upper bound on the pairwise outage probability (UBPOP). We start with SBCs.

\begin{proposition}
\label{proposition1}
For space block codes an upper bound on the pairwise outage probability is given as
\begin{equation}
\label{ch2:outagecutoffclosed}
\newcommand\norm[1]{\left\lVert#1\right\rVert}
p(q_{i,j} <\hat{\rho}_{i,j})=\frac{\gamma(N_{\rm{r}},-4N_0\norm{\Delta_{i,j}}_{\rm{F}}^{-2}\ln(q_{i,j}) )}{\Gamma(N_{\rm{r}})},
\end{equation}
where $q_{i,j}\geq0$ and $\gamma$ and $\Gamma$ are lower incomplete gamma and gamma functions, respectively.
\end{proposition}
\begin{proof}
In Appendix.~\ref{sec:appendixoutage}.
\end{proof}

Note that $q_{i,j}$ in general should be different for each pair since the distance of pairs in a STBC and correspondingly the effective SNR and their effect on union bounds may be different. However, finding $q_{i,j}$ needs complicated optimizations. Instead, we can find a relatively good value for $q_{i,j}$ as follows. Using (\ref{cutoff12}), the pairwise events $\{q_{i,j}<\hat{\rho}_{i,j}\}$ can be written as $\{R_0({\hat{\textbf{S}}_{i,j}|\textbf{H})}<R_{i,j}\}$ where $R_{i,j}=-\log_2(\frac{1}{2}+\frac{1}{2}q_{i,j})$ and $R_0(\hat{\textbf{S}}_{i,j}|\textbf{H})$ is the pairwise cutoff rate. Assuming the average spectral efficiency $R_{\rm{tot}}$ equals $R_{i,j}$ for all pairs of the STBC, $q_{i,j}$ can be found as $q=2^{1-R_{\rm{tot}}}-1$. Next, to make the upper bound manageable and find a relevant measure for designing SPM, we can approximate the UBPOP for STBCs as follows:
\begin{equation}
\label{ch2:outagecutoffclosedSTBC}
p(q_{i,j} <\hat{\rho}_{i,j})\approx\frac{\gamma(\frac{\mu_1^2}{\mu_2},-4N_0 \frac{\mu_1}{\mu_2} \ln(q_{i,j}) )}{\Gamma(\frac{\mu_1^2}{\mu_2})},
\end{equation}
where 
\begin{equation}
\label{moments}
  \mu_1= N_{\rm{r}}\sum_l\sum_{n_{\rm{t}}} |\delta_{l,n_{\rm{t}}}|^2, \quad
  \mu_2\approx N_{\rm{r}}\sum_{u=1}^{L} \sum_{v=1}^{L} |\sum_{n_{\rm{t}}} \delta_{u,n_{\rm{t}}}\delta_{v,n_{\rm{t}}}^*|^2.
\end{equation}
The approximation is derived in Appendix.~\ref{sec:appendixoutage2}.

Note that since we approximate the upper bound in (\ref{ch2:outagecutoffclosedSTBC}), it is not anymore an upper bound but it remains a relevant measure for designing SPM. This class of measure can be used to design the SPM for STBCs. However, for TVSBCs, we need to model TV sequences as well. In this case, we can model the TV sequence using the cutoff rate since it is a fast-changing parameter. An approximation of UBPOP for a $2\times N_{\rm{r}}$ TVSBC is
\begin{equation}
\label{ch2:outagecutoffclosedTVSBC}
p(q_{i,j} <\hat{\rho}_{i,j})\approx1-\text{Q}\Bigg(\frac{\ln(-4N_0\ln(q_{i,j}))+\ln(\frac{1}{\mu_1}+\frac{\mu_2}{\mu_1^3})}{\sqrt{\ln(1+\frac{\mu_2}{\mu_1^2})}}\Bigg),
\end{equation}
where $\text{Q}(.)$  is the Gaussian tail function. The first two moments $\mu_1$ and $\mu_2$ are given as
\begin{equation}
\label{firsttwomomnets}
\begin{split}
\mu_1= & N_{\rm{r}}\big( |\delta_{1,1}|^2+|\delta_{1,2}|^2\big) -2 \big(a_1|\delta_{1,1}\delta_{1,2}| E[\Omega] +a_2\big),  \\
\mu_2= & N_{\rm{r}}\big(|\delta_{1,1}|^4+|\delta_{1,2}|^4)+4a_1^2|\delta_{1,1}\delta_{1,2}|^2(N_{\rm{r}}-E[\Omega]^2) \\
&-4a_1N_{\rm{r}}\big( |\delta_{1,1}|^2+|\delta_{1,2}|^2\big)|\delta_{1,1}\delta_{1,2}| (E[|h_{1,n_{\rm{r}}}|^2\Omega]-E[\Omega]),
\end{split}
\end{equation}
 where $\Omega=|\sum_{n_{\rm{r}}}h_{1,n_{\rm{r}}}h_{2,n_{\rm{r}}}^*|$ and coefficients $a_1$ and $a_2$ are defined in (\ref{approximatelnI}). The approximation is derived in Appendix.~\ref{sec:appendixoutageSBCTV}.

The provided UBPOPs  can be used as a measure to construct the mapping. However, it turns out that the use of the determinant criterion and UBPOP for STBCs result in the same SPM. The same result is valid for SBCs when we use the corresponding UBPOP and the Frobenius norm, although for TVSBCs, the use of UBPOP results in an improvement compared to the Frobenius norm.

To overcome the shortcomings of these measures, we can analyze different conditions that may happen for pairwise measures in a slow fading channel and modify the labelling algorithm in \cite{Balogun2017}. The general form of the bounds derived in this section is $\textrm{Pr}(q_{i,j} <\hat{\rho}_{i,j})$. We expect that as the distance of two pairs increases, $\hat{\rho}_{i,j}$ decreases and correspondingly $-\textrm{Pr}(q_{i,j} <\hat{\rho}_{i,j})$ increases. However, due to a variety of reasons, e.g., the effect of pairwise measures on the sum cutoff rate or on the equality of the outage probability of different levels, $-\textrm{Pr}(q_{i,j} <\hat{\rho}_{i,j})$ may decrease with the distance. But it is unlikely that it happens when the difference in pairwise distances is relatively large. Motived by this explanation, we consider the same value for distances of two pairs of points if the relative difference of their distances is smaller than a threshold. The threshold is determined in a few iterations by constructing the mapping and the polar coded-modulation. In fact, we use the simulation-based design to find the best threshold. In practice, we observe that the threshold remains constant (typically around 0.1-0.2) for a wide range of SNRs. Also while this threshold is substantially effective on unoptimized SBCs, e.g., $\textbf{G}_{\rm{C}} $, it has a negligible effect on optimized SBCs, OSTBCs, and approximately universal STBCs. 

\begin{figure*}[!tb]
\centering   
 \includegraphics[width=0.7\textwidth]{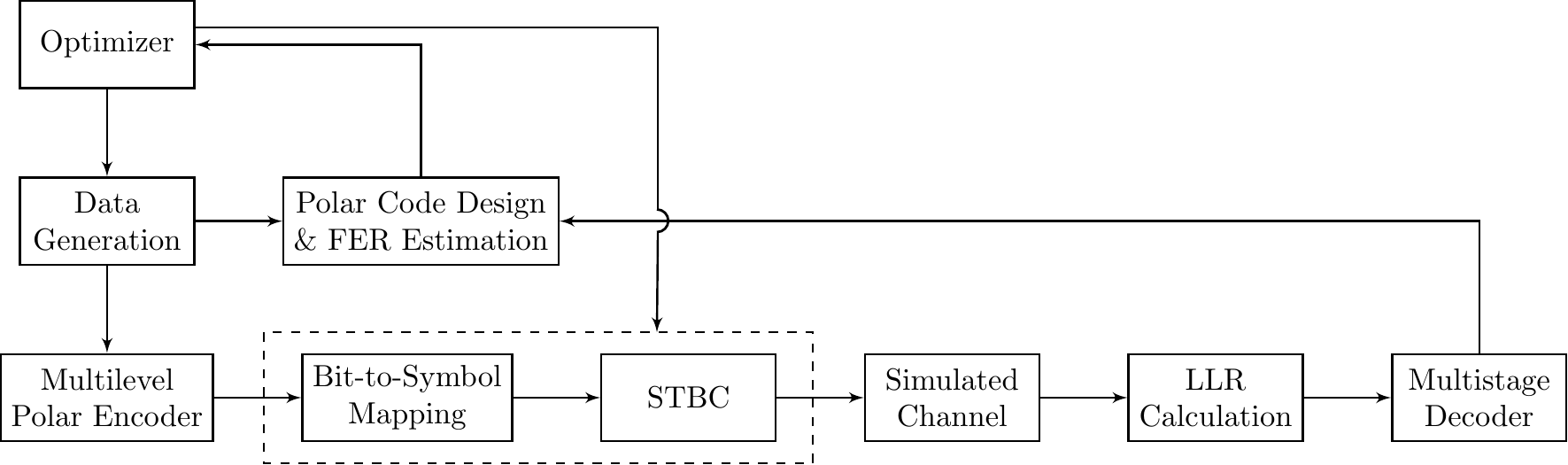}
\caption{Block diagram of the joint optimization of polar codes and  STBCs.}
\label{fig:blockdiagram}
\end{figure*}

\subsection{Design Procedure for the MLPCM}
\label{codedModulationconstruction}

To design a MLPCM for STBCs, the multidimensional SPM is generated for a STBC using the algorithm reviewed in Section~\ref{labelling}. Then, the multilevel polar code is designed for the constructed set-partitioned STBC scheme by using the code design method proposed in \cite{Balogun2017}. In this paper, MLPCM is designed at the minimum SNR such that a target FER of $0.01$ is achieved.

\begin {table*}[!b]
  \caption{Summary of compared signal design methods and the corresponding figure numbers.}
\label{tbl1}
\begin{center}
  \begin{tabular}{ | c | c | c | c | c | c | c | c | c | c| c| c|}
    \hline
    STBC &  $\textbf{G}_{\rm{A}}$ &  $\textbf{G}_{\rm{B}}$ & $\textbf{G}_{\rm{D}}$& $\textbf{G}_{\rm{D}}$, TV & $\textbf{G}_{\rm{E}}$ & $\textbf{G}_{\rm{F}}$ & $\textbf{G}_{\rm{F}}$, TV \\ \hline
    Rank and Determinant Criteria (RD) & 4 & 4  & & &  & & \\ \hline
    Mutual Information (MI) &  & 4 & & & & & \\ \hline
    Outage Probability (OP) & & 4 & 6& 6&5  & & \\ \hline
    Joint Design of FEC and STBC (JD) &  & 4& & 6&  &7&7 \\ \hline
  \end{tabular}
\end{center}
\end {table*}

\section{Joint Optimization of FEC Code and STBC}
\label{sec:JointOptimization}

The outage probability is a good criterion for designing the modulation (or STBC) for FEC codes that perform close to the outage. However, when there is a gap between the performance of an FEC code and the outage, the outage probability is not necessarily a good measure. This gap exists for all currently available finite length FEC codes, even in the presence of powerful decoders, and motivates the research on the joint optimization of STBCs and practical FEC codes.

To optimize concatenated schemes, bounds on the performance of the concatenated FEC code and STBC are needed. For cases such as BICM, if the LLR distribution is nearly Gaussian, general bounds on the performance of the concatenated schemes can be derived \cite{Gresset2008}. However, for MLC/MSD in which the LLR distribution of each subset within each level is different, deriving any closed-form bound on the performance of the system is difficult. Furthermore, optimizing the sum cutoff rate bound (see \cite{Arikan2006}) is numerically intensive. Instead, the simulation-based design method, explained in Section~\ref{sec:codedesign}, can determine the information set and the FER of the limited length MLPCM. In fact, simulation-based design plays the role of a bound on the performance of the system since it can approximate the FER. Thus, using the simulation-based method, joint optimization of limited length FEC codes and STBCs is possible.

In simulation-based polar code design, given a specific code rate and SNR value, the information set for MLPCM is chosen and the FER of the polar coded-modulation is determined \cite{Balogun2017}.  To jointly design the polar code and STBC, the best STBC matched to the polar code structure should be determined.  To this end, the optimizer generates a set of parameters of a specific STBC, e.g., $\textbf{G}_{\rm{B}} $. Then, for each combination of parameters of the STBC, the set-partitioning algorithm described in Section \ref{labelling} is applied to find a good SPM for the generated STBC and polar code design procedure for the new set-partitioned STBC is repeated. Finally, the best match of the information set, SPM, and parameters of a specific STBC corresponding to the lowest FER is chosen. The modified PSO, described in Section~\ref{sec:STBCoutage}, is used for the joint optimization. The block diagram of the joint optimization method is shown in Fig.~\ref{fig:blockdiagram}.

To compare the performance of different design methods in Section~\ref{sec:numresult}, we optimized a few codes using the joint design algorithm at a minimum SNR corresponding to a FER of $0.01$. The  parameters of theses codes are as follows; for $\textbf{G}_{\rm{B}} $, $\alpha_1=\beta_1=0.765$, $\alpha_2=-0.265+0.587\jmath$, and $\beta_2=0.587+0.265\jmath$; for TV $\textbf{G}_{\rm{D}} $, $\alpha_1=\alpha_2=0.5$ and $\beta_2=\jmath\beta_1=-0.410 + 0.287\jmath$; for $\textbf{G}_{\rm{F}}$, $\beta_1=0.358 + 0.358\jmath$, $\beta_2=-0.358-0.358\jmath$ and $\beta_3=-0.506$; and for TV $\textbf{G}_{\rm{F}}$, $\beta_1=0.359+0.359\jmath$, $\beta_2=-0.359-0.359\jmath$, and $\beta_3=-0.507$.

Note that the joint optimization method is less complex than the outage probability optimization for short to moderate length codes since the relatively precise numerical evaluation of the outage probability in (\ref{ch2:outage}) is expensive. However, for long codes, the joint optimization method is expensive and it turns out that the outage probability optimization is cheaper. Furthermore, in online optimization when the statistics of the communication channel are not known, both optimization based on the outage probability  and the joint optimization methods can be used. But employing the joint optimization is more affordable since it does not need an additional software as the polar encoder and decoder are embedded in the system.

\section{Numerical Results and Discussions}
\label{sec:numresult}

\begin {table*}[t!]
  \caption{Summary of compared SPM design rules and the corresponding figure numbers.}
\label{tbl2}
\begin{center}
  \begin{tabular}{ | c | c | c |c | c | c | c | c | c | c| c| c|}
    \hline
    STBC & Alamouti & $\textbf{G}_{\rm{A}}$ &  $\textbf{G}_{\rm{B}}$ & $\textbf{G}_{\rm{C}}$ & $\textbf{G}_{\rm{C}}$, TV & $\textbf{G}_{\rm{D}}$& $\textbf{G}_{\rm{D}}$, TV & $\textbf{G}_{\rm{E}}$ & $\textbf{G}_{\rm{F}}$ & $\textbf{G}_{\rm{F}}$, TV \\ \hline
    Frobenius norm (FN) & 4 & 4 & 4 & 5,7 &5& & &  & 7&7\\ \hline
   Modified Frobenius norm (MFN) & & & & 5,6,7 & 5,7& 6&6 & 5 & & \\ \hline
    UBPOP & & & &  & 5,6 & & 6&  &  & \\ \hline
  \end{tabular}
\end{center}
\end {table*}

In this section, the performance of the outage optimized and the joint optimized MLPCM schemes with STBCs and SBCs are compared with MLPCM designed using the rank and determinant criteria, and the mutual information. STBCs are designed based on the rank and determinant criteria, referred to as the RD Method, the mutual information, referred to as the MI Method, the outage probability, referred to as the OP Method, and the joint design of FEC codes and STBCs referred to as the JD Method. The parameters of the optimized codes are mentioned in Sections~\ref{sec:STBCoutage} and \ref{sec:JointOptimization} for the OP and JD  Methods, respectively. For designing MLPCM, the construction method in Section~\ref{codedModulationconstruction} is employed. We also compare MLPCM with bit-interleaved turbo coded-modulation (BITCM). For BITCM, the BCJR decoder is used with 20 iterations of the turbo decoder. The turbo code in BITCM is LTE's turbo code \cite{ETSITS136212}. The optimized parameters are mentioned in Section~\ref{sec:STBCoutage}. The polar decoder in all cases is SCD. The SC list decoder, introduced in \cite{Tal2015}, is also tried but no curve has been shown. For all SCD curves, SCLD slightly improves the performance, e.g., around 0.2 dB. Unless otherwise stated, we use a $2 \times 2$ antenna configuration. The constellation used for Matrices A, B, E, and F is QPSK and for Matrix C and D is 16-QAM. A summary of compared signal design methods and the corresponding figure numbers are mentioned in Table.~\ref{tbl1}.

The performance of the BITCM and MLPCM with Matrices A, B, and the Alamouti STBC for 2 bpcu with $R_{\rm{tot}}=1/2$ and $N_{\rm{tot}}=512$ bits are compared in Fig. \ref{fig:result1}. MLPCM with Alamouti STBC scheme is constructed using an MLC with 4 levels for a 16-QAM constellation with the total code rate of $1/2$. The BITCM with $\textbf{G}_{\rm{A}} $ is 0.8 dB worse than the MLPCM with $\textbf{G}_{\rm{A}} $ at a FER of $0.01$. Note that the BCJR-based turbo decoder with 20 iterations is far more complex than the SCD. In the set of curves provided for MLPCM with $\textbf{G}_{\rm{B}} $, the STBC designed using the MI Method is 1.2 dB worse than the RD Method since the mutual information is not an appropriate measure for the slow fading channel.  Furthermore, FERs of STBCs designed using the OP and ID Methods are 0.1 dB and 0.4 dB lower than the one with the RD Method, respectively. Thus, as expected, the joint optimization can improve the performance even more than the optimization based on the outage probability.

\begin{figure}[!b]
\centering   
 \includegraphics[width=0.47\textwidth]{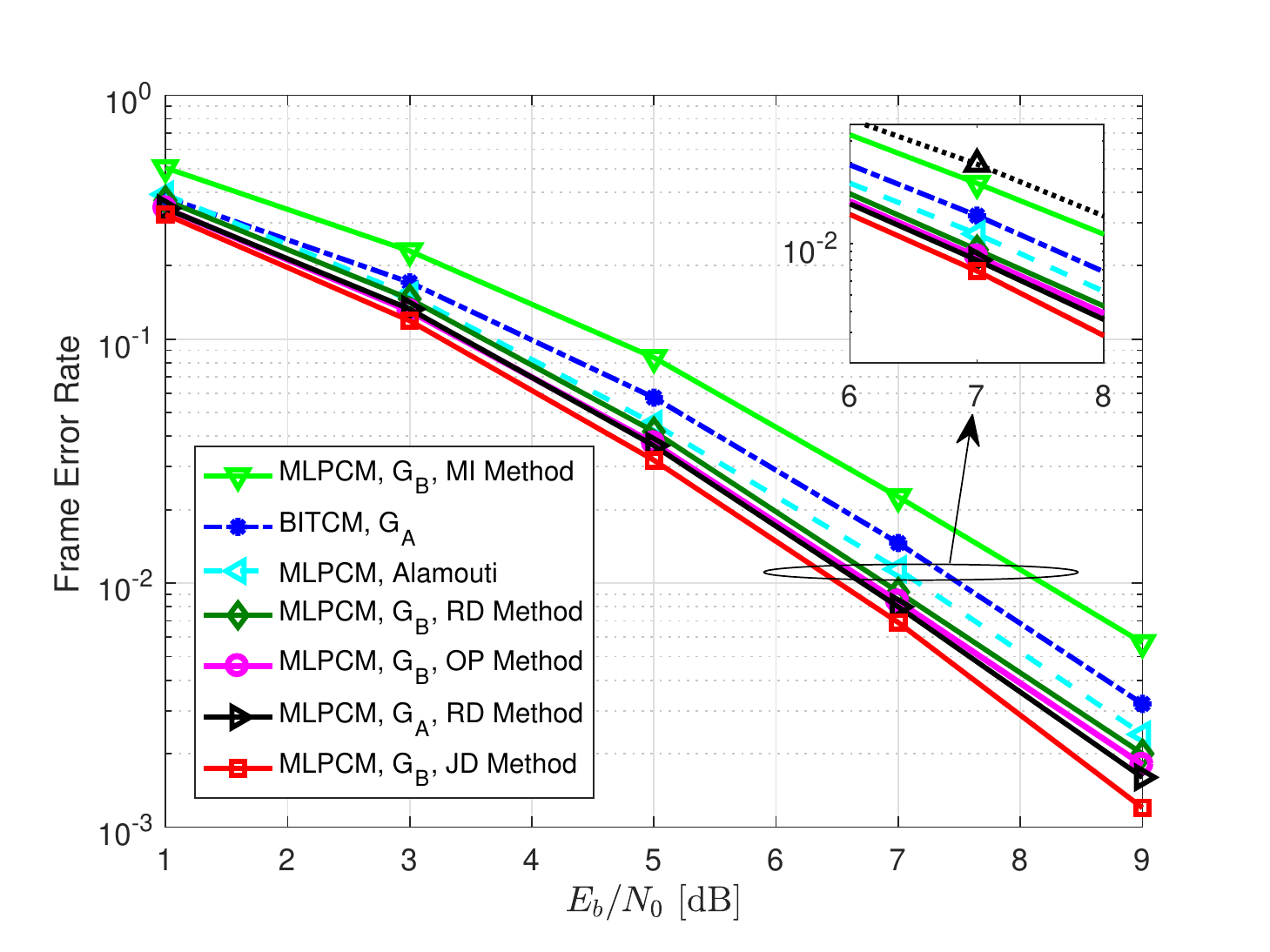}
\caption{FER comparison of BITCM and MLPCM with $\textbf{G}_{\rm{A}} $, and MLPCM with $\textbf{G}_{\rm{B}} $ designed using the rank and determinant criteria (RD), the mutual information (MU), the outage probability (OP), and the joint design (JD) Methods, and the Alamouti code for 2 bpcu and $N_{\rm{tot}}=512$.}
\label{fig:result1}
\end{figure}

In Fig. \ref{fig:result1} for all curves, we used the SPM designed based on the determinant criterion since other measures do not generate better SPMs or the improvement is negligible. In Fig. \ref{fig:result2}, we evaluate the effect of labelling rules including the Frobenius norm, referred to as the FN Rule, the Frobenius norm using the modified labeling algorithm presented at the end of Section \ref{labelling}, referred to as the MFN Rule, and the design based on the UBPOP, referred to as the UBPOP Rule, for 4 bpcu. A summary of the compared SPM design rules and the corresponding figure numbers are mentioned in Table.~\ref{tbl2}.   For all curves, $N_{\rm{tot}}=1024$ bits and $R_{\rm{tot}}=1/2$. We observe that for the MLPCM with $\textbf{G}_{\rm{C}} $, the SPM designed using the MFN Rule, improves the FER by about 2 dB in comparison to the FN Rule at a FER of $0.01$. Moreover, for TV $\textbf{G}_{\rm{C}} $, SPMs designed using the MFN and UBPOP Rules work 1.2 and 1.7 dB better than the SPM generated using the FN Rule, respectively.

\begin{figure}[!h]
\centering   
 \includegraphics[width=0.47\textwidth]{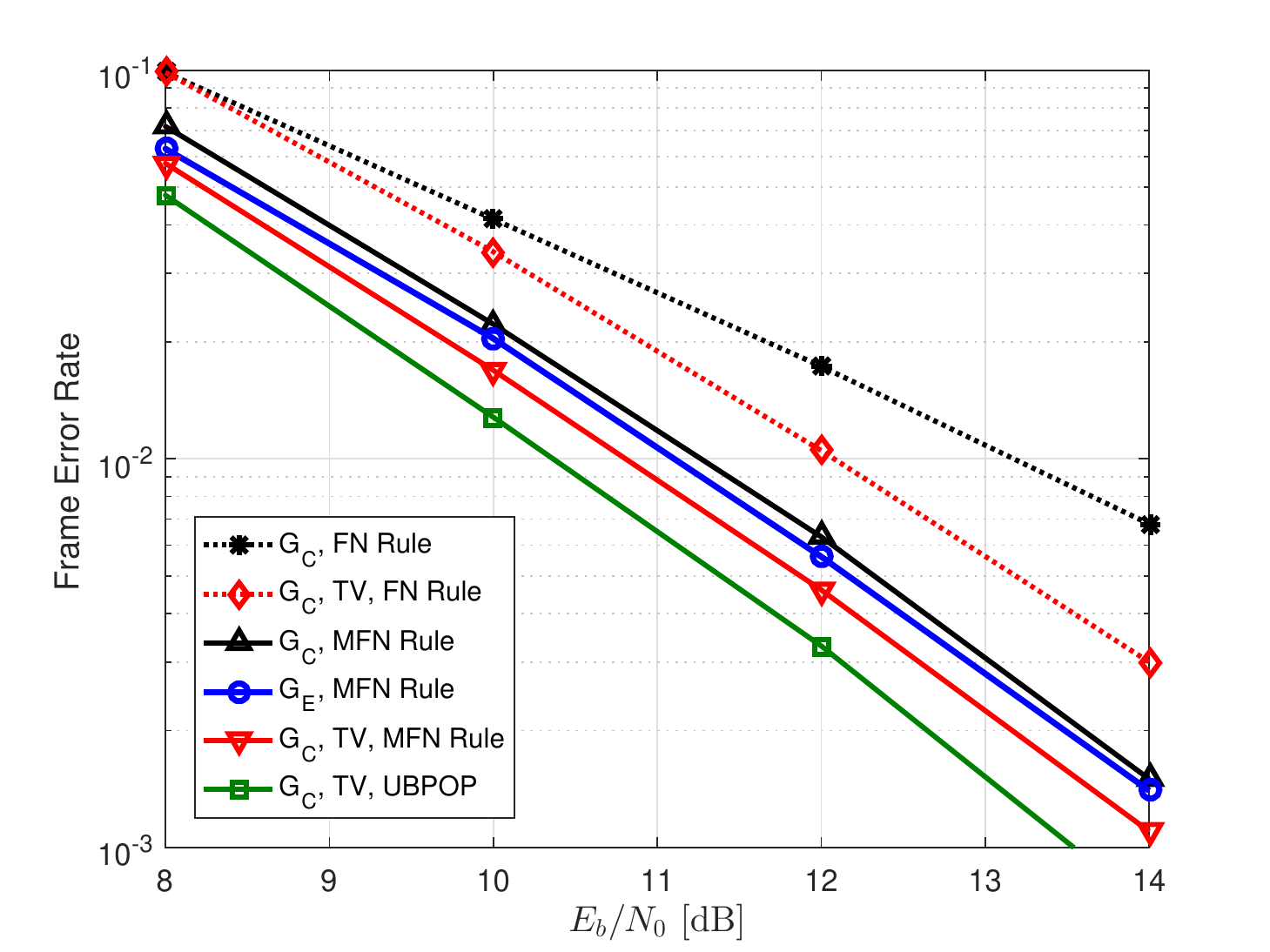}
\caption{FER comparison of MLPCM with $\textbf{G}_{\rm{C}} $ designed using the Frobenius norm (FN), modified Frobenius norm (MFN), and UBPOP rules for SPM generation and $\textbf{G}_{\rm{E}} $ designed using the OP Method, all for 4 bpcu and $N_{\rm{tot}}=1024$.}
\label{fig:result2}
\end{figure}

In Fig.~\ref{fig:result3}, we provided the comparison of  $\textbf{G}_{\rm{D}} $ and $\textbf{G}_{\rm{C}} $ at 7.2 bpcu for the total code length of $128$ bits and $R_{\rm{tot}}=9/10$.  The results indicate that at a FER of $0.001$, the MLPCM with TV $\textbf{G}_{\rm{D}} $ designed using the OP Method outperforms MLPCM with $\textbf{G}_{\rm{C}} $ and TV $\textbf{G}_{\rm{C}} $ by 0.8 dB and 0.6 dB, respectively. Furthermore, for MLPCM with TV $\textbf{G}_{\rm{D}} $, optimization using the JD Method provides 0.2 dB improvement over the OP Method.  It is clear that the joint optimization of polar codes and STBCs for all code rates can slightly improve the performance in comparison to the optimization based on the outage probability. Also, by comparing  Fig.~\ref{fig:result3} with Fig.~\ref{fig:b}, we realize that the order of curves in terms of the FER is the same as that of the outage probability curves.

\begin{figure}[!h]
\centering    
 \includegraphics[width=0.47\textwidth]{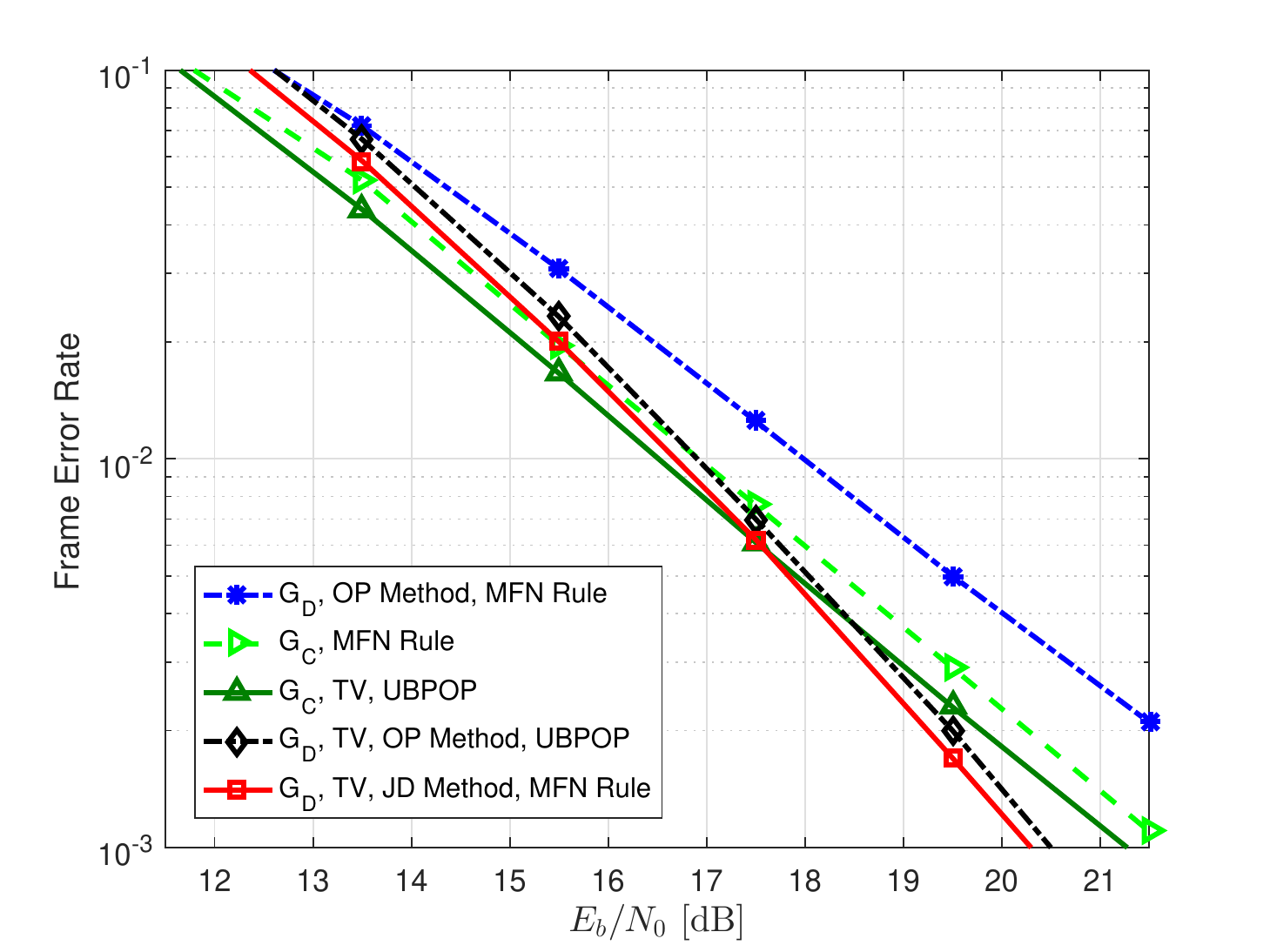}
\caption{a) FER comparison of MLPCM with $\textbf{G}_{\rm{C}} $ and $\textbf{G}_{\rm{D}} $ designed using the OP and JD Methods, all for 7.2 bpcu and $N_{\rm{tot}}=128$.}
\label{fig:result3}
\end{figure}

Finally, the comparison of MLPCM with $\textbf{G}_{\rm{C}} $ and $\textbf{G}_{\rm{F}}$ for 6 bpcu by setting $N_{\rm{tot}}=256$, $R_{\rm{tot}}=1/2$ and $N_{\rm{t}}=N_{\rm{r}}=3$ is shown in Fig. \ref{fig:result4}. We observe that for MLPCM with $\textbf{G}_{\rm{C}} $, employing the MFN Rule improves the performance by 0.3 dB in comparison to the FN Rule at a FER of $0.01$. This improvement is less than what we have observed for $N_{\rm{r}}=2$ antennas in Fig.~\ref{fig:result2} since as $N_{\rm{r}}$ increases, the diversity order increases and eventually the Frobenius norm, as the measure for designing SPM, becomes optimal. In addition, MLPCM with $\textbf{G}_{\rm{F}}$  and TV $\textbf{G}_{\rm{F}}$ optimized using the JD Method work 0.5 and 0.6 dB better than MLPCM with $\textbf{G}_{\rm{C}} $  and TV $\textbf{G}_{\rm{C}} $, respectively.

\begin{figure}[!h]
\centering    
 \includegraphics[width=0.47\textwidth]{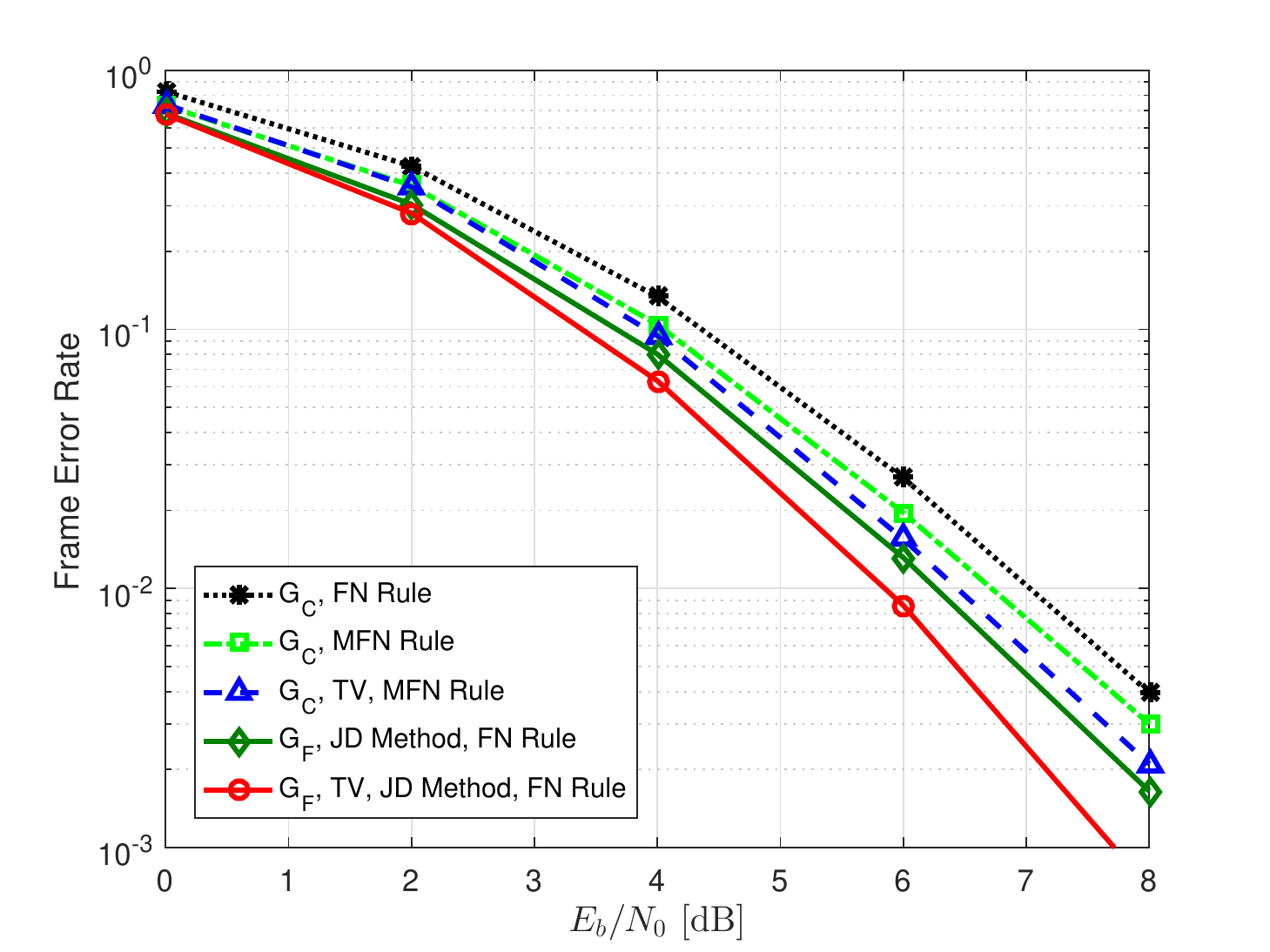}
\caption{FER comparison of MLPCM with $\textbf{G}_{\rm{C}} $ and $\textbf{G}_{\rm{F}} $ designed using the JD Method, all for 6 bpcu and $N_{\rm{tot}}=256$ in a $3\times 3$ antenna configuration.}
\label{fig:result4}
\end{figure}

\section{Conclusion}
\label{sec:Conclusion}

In this paper, we improved the space-time signal design for multilevel polar coding as a low complexity power-efficient scheme for slow broadcast channels. To do so, we first optimized STBCs for MLPCM by minimizing the outage probability at a target outage or SNR. The method includes limiting the number of free parameter of STBCs by using power and shaping conditions and employing a modified PSO algorithm to find other parameters. In addition, we showed that the outage probability of all levels of MLC/MSD should be optimally equal under certain conditions and based on that we proposed an outage rule for determining component code rates. Furthermore, to design SPM we derived an upper bound on the pairwise outage probability that can model both slow- and fast-changing parameters of the system and the channel. We showed employing this bound to generate SPM for TVSBCs can substantially improve the performance up to 1.7 dB compared to the Frobenius norm. Due to the similarity of SPMs generated using the derived bound to those generated using the Frobenius norm for SBCs, we further proposed an algorithm to modify pairwise measures that decreased the FER up to 2 dB.

 We also proposed a novel approach to jointly optimize multilevel polar codes and STBCs. For the joint optimization of polar code and STBCs, here, we first change the parameters of a STBC to create a new STBC; then we generate a new SPM using a set-partitioning algorithm and repeat the code design procedure for the new STBC until we find an information set and the corresponding parameters of the STBC that minimize the FER jointly. The modified PSO can be employed for the joint optimization as well. The numerical results show an improvement compared to STBCs, SBCs and TVSBCs designed using the rank and determinant criteria. For future work, algorithms could be extended to a variety of alternative channel models, including correlated and multi-path fading channels.
    
\appendices
\section{The proof of the bound (\ref{ch2:outagecutoffclosed})}
\label{sec:appendixoutage}

We begin with simplifying the Bhattacharyya coefficient for the space-time signal model in (\ref{recievedsignal}). By substituting (\ref{prob}) and using \cite[Lemma.~1]{Dogandzic2003}, (\ref{pairwise1}) can be simplified as 
\begin{equation}
\label{cutoff124}
\hat{\rho}_{i,j}=\text{exp}\Big(\frac{-||(\textbf{S}_i-\textbf{S}_j)\textbf{H}||_{\rm{F}}^2}{4N_0}   \Big).
\end{equation}

Using (\ref{pairwise1}), the upper bound on the outage probability denoted as $p(q_{i,j} <\hat{\rho}_{i,j})$  can be written as 
\begin{equation}
\label{cutoff125}
p\Bigg( q_{i,j}<\text{exp}\Big(\frac{-||\Delta_{i,j}\textbf{H}||_{\rm{F}}^2}{4N_0}\Big)\Bigg).
\end{equation}
Taking $\ln$ of $q_{i,j}$ and $\hat{\rho}_{i,j}$, (\ref{cutoff125}) can be simplified as 
\begin{equation}
\label{cutoff126}
p\Big(||\Delta_{i,j}\textbf{H}]||_{\rm{F}}^2<-4N_0\ln(q_{i,j})\Big).
\end{equation}
In (\ref{cutoff126}), the left hand side of the inequality can be written as
\begin{equation}
\label{cutoff127}
||\Delta_{i,j}\textbf{H}||_{\rm{F}}^2=\sum_{n_{\rm{r}}}\sum_l\Big|\sum_{n_{\rm{t}}}\delta_{l,n_{\rm{t}}}h_{n_{\rm{t}},n_{\rm{r}}}\Big|^2.
\end{equation}
 The term $\Lambda$ inside of $\xi_{l,n_{\rm{r}}}=|\sum_{n_{\rm{t}}}\delta_{l,n_{\rm{t}}}h_{n_{\rm{t}},n_{\rm{r}}}|^2=|\Lambda|^2$ in (\ref{cutoff127}) can be algebraically extended to the real and imaginary parts. Each real and imaginary part is the weighted sum of terms including real or imaginary part of an element of $\textbf{H}$ denoted as $h_{n_{\rm{t}},n_{\rm{r}}}$. Therefore, both real and imaginary parts are the sum of Gaussian random variables and consequently have Gaussian distribution. Since the real and imaginary parts of $\xi_{l,n_{\rm{r}}}$ are Gaussian random variables and their correlation turns out to be zero, their sum is distributed according to a central chi-squared distribution with two degrees of freedom. In case of SBCs, (\ref{cutoff127}) is simplified to $\sum_{n_{\rm{r}}}|\sum_{n_{\rm{t}}}\delta_{1,n_{\rm{t}}}h_{n_{\rm{t}},n_{\rm{r}}}|^2$ in which for different receive antennas, the term  $\xi_{1,n_{\rm{r}}}$ is independent of others. Therefore, the distribution of sum of $\xi_{1,n_{\rm{r}}}$ is the sum of chi-squared distributions with two degrees of freedom which in turn is a chi-squared distribution with $2N_{\rm{r}}$ degrees of freedom and $\mu_2=\sum_{n_{\rm{r}}}\mu_{2,n_{\rm{r}}}$ where $\mu_{2,n_{\rm{r}}}$ is the variance for each receive antenna. The variance of the each real and imaginary part of $\Lambda$ in $\xi_{1,n_{\rm{r}}}$ is $\mu_2=0.5\sum_{n_{\rm{t}}}|\delta_{1,n_{\rm{t}}}|^2$. From (\ref{cutoff126}), the outage is the CDF of this distribution expressed in (\ref{ch2:outagecutoffclosed}).

\section{Derivation of the Approximation (\ref{ch2:outagecutoffclosedSTBC})}
\label{sec:appendixoutage2}

In this section, we use the definitions in Appendix~\ref{sec:appendixoutage}. In case of STBCs, for each receive antenna, (\ref{cutoff127}) can be simplified as $\sum_{l}\xi_{l,n_{\rm{r}}}$ which includes $L$ correlated terms. Since $\xi_{l,n_{\rm{r}}}$ has a gamma distribution, for $L=2$, $\sum_{l}\xi_{l,n_{\rm{r}}}$ is distributed according to a type I McKay distribution \cite{Holm2004}. For $L>2$, the general expression $\sum_{l}\xi_{l,n_{\rm{r}}}$ is the sum of  correlated gamma random variables and is distributed according to a complex infinite power series \cite{Alouini2001}. As a computationally cheap yet accurate approximation, the sum of  correlated gamma random variables can be accurately approximated as a gamma random variable by matching the first two moments \cite{Feng2016}. In a few straightforward steps, these moments can be derived  as (\ref{moments}) and the CDF is expressed in (\ref{ch2:outagecutoffclosedSTBC}).

\section{Derivation of the Approximation (\ref{ch2:outagecutoffclosedTVSBC})}
\label{sec:appendixoutageSBCTV}
 
The Bhattacharyya coefficient in (\ref{cutoff124})  for one receive antenna by considering TV sequences can be expressed as 
\begin{equation}
\label{cutoff128}
\rho(\textbf{S}_i,\textbf{S}_j|\textbf{H},\boldsymbol{\theta})={ \rm{exp}}\Big(- {\frac{1}{4 {N}_0}} \Big|\sum_{n_{\rm{t}}}\delta_{1,n_{\rm{t}}}h_{n_{\rm{t}},n_{\rm{r}}}e^{\jmath\theta_{n_{\rm{t}}}}\Big|^2\Big).
\end{equation}
Due to the fast-changing nature of TV sequences and assuming the code length tends to infinity, the integral over $\rho(\textbf{S}_i,\textbf{S}_j|\textbf{H},\boldsymbol{\theta})$ can be taken as
\begin{equation}
\label{cutoff200}
\begin{split}
&\boldsymbol{\int_{\theta}} \text{exp}\Big(-  \frac{1}{4N_0} \Big|\sum_{n_{\rm{t}}}\delta_{t,n_{\rm{t}}}h_{n_{\rm{t}},n_{\rm{r}}}e^{j\theta_{n_{\rm{t}}}}\Big|^2\Big) \frac{d\boldsymbol{\theta}}{(2\pi)^{N_{\rm{t}}}} \\
&= \text{exp}\big(- \frac{1}{4N_0}\sum_{n_{\rm{t}}}|\delta_{1,n_{\rm{t}}}h_{n_{\rm{t}},n_{\rm{r}}}|^2  \big) \boldsymbol{\int_{\theta}}\frac{\text{exp}(- F_{n_{\rm{r}}}(\Delta_{i,j}))}{(2\pi)^{N_{\rm{t}}}} d\boldsymbol{\theta},
\end{split}
\end{equation}
where
\begin{equation}
\label{cutoff201}
F_{n_{\rm{r}}}(\Delta_{i,j})=\frac{2}{4N_0}{\text{Re}}\{\sum_{u=1}^{N_{\rm{t}}}\sum_{v=1}^{N_{\rm{t}}}\delta_{1,u}\delta_{1,v}h_{u,n_{\rm{r}}}h_{v,n_{\rm{r}}}e^{-\jmath\hat{\theta}}\},
\end{equation}
in which $\hat{\theta}=\theta_u-\theta_v$ and $(2\pi)^{-N_{\rm{t}}}$ is the density function of $\boldsymbol{\theta}$ which is uniformly distributed in range $[0,2\pi]$. Taking the general form of the integral $\boldsymbol{\int_{\theta}}\text{exp}(- \frac{1}{4N_0}F_{n_{\rm{r}}}(\Delta_{i,j}))$ in (\ref{cutoff200}) is difficult. However, for $N_{\rm{t}}=2$, it can be simplified as
\begin{equation}
\label{cutoff202}
I_0( \frac{1}{2N_0}|\delta_{1,1}\delta_{1,2}|),
\end{equation}
where $I_0$ is the modified Bessel function of the first kind with zero order. To estimate the upper bound on the outage probability, we take the logarithm of the  Bhattacharyya coefficient in  (\ref{cutoff126}). Similarly, for $N_{\rm{t}}=2$ and $N_{\rm{r}}=1$ by taking the logarithm, the left hand side of (\ref{cutoff126}) is simplified as 
\begin{equation}
\label{cutoff203}
\frac{1}{4N_0}\underbrace{\sum_{n_{\rm{t}}}|\delta_{1,n_{\rm{t}}}h_{n_{\rm{t}},n_{\rm{r}}}|^2}_{\textrm{Part 1}}- \underbrace{\ln\big(I_0(\frac{1}{2N_0}|\prod_{n_{\rm{t}}}\delta_{1,n_{\rm{t}}}h_{n_{\rm{t}},n_{\rm{r}}}|)}_{\textrm{Part 2}}\big).
\end{equation}

In (\ref{cutoff203}), $\ln\big(I_0(x)\big)$ can be estimated using piecewise linear approximation given as
\begin{equation}
\label{approximatelnI}
\ln\big(I_0(x)\big)\approx  a_1x+a_2= \small
\begin{cases} 
0.12x&  0<x \leq 0.5, \\  
0.35x-0.12&  0.5<x\leq 1, \\ 
0.59x-0.37& 1<x \leq 2, \\
0.8x-0.81& 2<x \leq 4, \\
0.92-1.3&  x \geq 4. \\
\end{cases}
\end{equation}
Note that for $x>4$, the linear approximation remains valid for a wide range and thus, for designing a SPM, $\delta_{1,n_{\rm{t}}}$ or the signal energy can be adjusted to lie in this region.

Finding an exact distribution for (\ref{cutoff203}), as discussed in \cite{khoshnevis2018}, is difficult. Alternatively, to find a close distribution for a given pair of STBC points, the empirical distribution in (\ref{cutoff203}) is generated and compared with a large number of classes of distributions and the one with the lowest Kullback-Leibler divergence is chosen. To find a distribution that works well on average for all points of the STBC, the average Kullback-Leibler divergence over all STBC points is used. Thus, the measure to choose the distribution is given as
\begin{equation}
\label{approximatedoubleRayleigh}
\frac{1}{2^{2B}}\sum_u \sum_v  \sum_{i} \textrm{P}_i(\textbf{S}_u,\textbf{S}_v) \log_2\bigg(\frac{\textrm{P}_i(\textbf{S}_u,\textbf{S}_v)}{\textrm{Q}_i(\textbf{S}_u,\textbf{S}_v)}\bigg),
\end{equation}
where  $\textrm{P}_i(\textbf{S}_u,\textbf{S}_v)$ and $\textrm{Q}_i(\textbf{S}_u,\textbf{S}_v)$ are the $i^{\rm{th}}$ elements of the vector of probability values of the tested and empirical distributions, respectively. After testing a large number of classes of one and two parameters distributions, the log-normal distribution is chosen to model the distribution in (\ref{cutoff203}). To find the parameters of the under-test distribution, we match the first two moments. Note that although three parameter distributions may better model the empirical data, matching their first two moments may be difficult.  For $N_{\rm{r}}>1$,  Part 1 in (\ref{cutoff203}) changes to $\frac{1}{4N_0}\sum_{n_{\rm{r}}}\sum_{n_{\rm{t}}}|\delta_{1,n_{\rm{t}}}h_{n_{\rm{t}},n_{\rm{r}}}|^2$ and Part 2 should be extended as $- \ln(I_0(\frac{1}{2N_0}|\sum_{n_{\rm{r}}} \delta_{1,1}h_{1,n_{\rm{r}}}\delta_{1,2}^*h_{2,n_{\rm{r}}}^*|) )$. The matched CDF and the first two moments are presented in (\ref{ch2:outagecutoffclosedTVSBC}) and (\ref{firsttwomomnets}), respectively.  The empirical and the matched log-normal CDFs are depicted in Fig.~\ref{fig:CDFcompare} for two SNRs.

\begin{figure}[h!]
\center
 \includegraphics[width=0.47\textwidth]{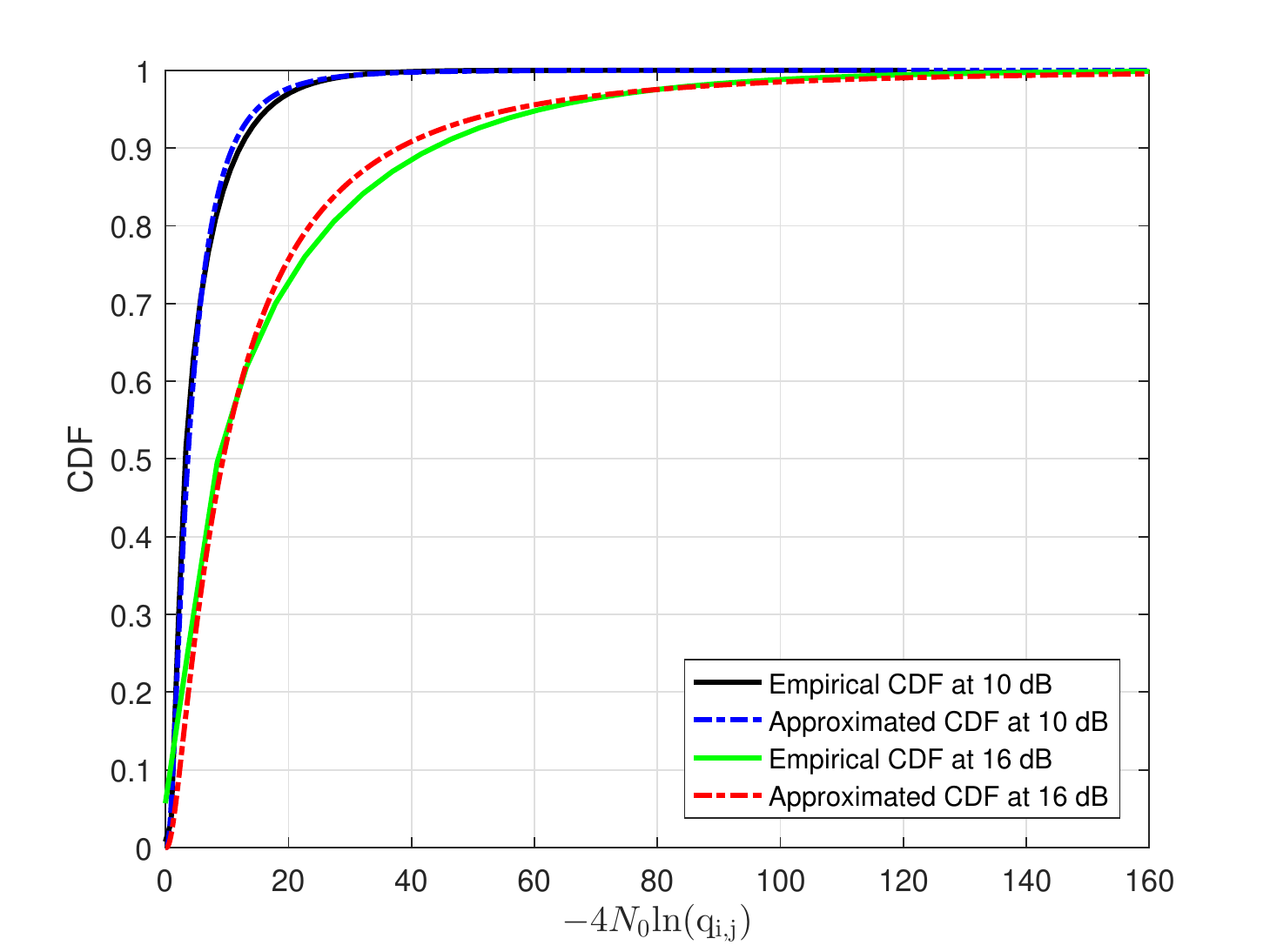}
\caption{Comparison of the empirical CDF and the approximated CDF for a given pair of TVSBC points at two different SNRs and $N_{\rm{r}}=1$.}
\label{fig:CDFcompare}
\end{figure}

\ifCLASSOPTIONcaptionsoff
  \newpage
\fi

\end{document}